\documentclass[journal]{IEEEtran}
\usepackage{epsfig,amsmath,amssymb,amsfonts,amsthm}
\usepackage{wasysym}
\usepackage{algorithmic}
\usepackage{algorithm}
\usepackage{subfloat}
\usepackage[acronym]{glossaries}
\usepackage{multirow}
\usepackage{caption}
\usepackage{subcaption}
\usepackage{txfonts}
\usepackage{soul}
\usepackage{cite}
\usepackage{blkarray}
\usepackage{nicematrix}
\usepackage{slashbox}

\usepackage{tikz}
\usetikzlibrary{arrows}

\NiceMatrixOptions{
code-for-first-row = \color{blue} ,
code-for-last-row = \color{blue} ,
code-for-first-col = \color{blue} ,
code-for-last-col = \color{blue}
}

\newtheorem{theorem}{Theorem}[section]
\newtheorem{lemma}[theorem]{Lemma}
\newtheorem{prop}[theorem]{Proposition}

\newtheorem{remark}{Remark}

\newtheorem*{klein}{Kleinrock's Observation}
\newtheorem*{klein_rv}{Revised Observation}
\newtheorem*{conj}{Conjecture}

\ifCLASSINFOpdf
\else
\fi
\begin{document}
%
\title{End-to-End Delay Approximation in Packet-Switched Networks}
%
%
%

\author{Yu~Chen
\thanks{Y. Chen is with the National Engineering Laboratory for Mobile
Network Technologies, Beijing University of Posts and Telecommunications,
Beijing, 100086, China e-mail: {yu.chen}@bupt.edu.cn.}
}

\maketitle

\begin{abstract}

In this paper, I develop a generalized method to approximate end-to-end delay (average delay, jitter and density functions) in packet-switched networks of any size under 1) Kleinrock's
independence assumption (KIA) and 2)
when packet lengths are kept unchanged
when they traverse from node to node in a network, which is an Alternative to Kleinrock's
independence assumption (AKIA). I introduce a new phase-type distribution $C(\mathbf{p},\boldsymbol \theta)$;
and then use results from the network flow
theory and queueing theory to show that
the end-to-end delay in PSNs under KIA and AKIA are two different random variables approximately described by $C(\mathbf{p},\boldsymbol \theta)$. When PSNs have AKIA, I show from simulation that the method under AKIA significantly reduces end-to-end delay approximation errors and provides close approximation compared with the method under KIA.

\end{abstract}

\begin{IEEEkeywords}
Phase-type distribution, packet-switched networks, jitter, delay distribution, Kleinrock's independence assumption.
\end{IEEEkeywords}

%
\IEEEpeerreviewmaketitle

\section{Introduction}
\IEEEPARstart{B}{ack} in 1961, Kleinrock used queueing theory to analyze message flows in message-switched communication networks (MSCNs)\footnote{It is
unnecessary to differentiate the terms 1) message and packet, and 2) message switching and packet switching from the analysis point of view. In this work, I use the terms ``packet'' and ``packet switching'' because packet switching dominates today's data networks.}. In his early work,
he raised a question about communication nets \cite{Kleinrock1961}:

\noindent \textbf{Kleinrock's question:} \emph{What is the probability density distribution for the total
time lapse between the initiation and reception of a message between any two nodes? }

\noindent and made his classic independence assumption \cite{Kleinrock1962}:

\noindent \textbf{Kleinrock's independence assumption (KIA):}  \emph{Each time that a message is
received at a node within the network, a new length is chosen for this packet
independently from an exponential distribution.}


%

KIA is widely used till today when designing or optimizing communication networks (see references in \cite{QiongLi2001,Sappidi2013,Teng2019,Schulz2020,Choi2020}). However, this assumption does not meet actual situations most of the time because packet lengths are usually unchanged as they pass through networks \cite{Kleinrock1976}. Here, we have another more realistic assumption:

\noindent \textbf{Alternative to Kleinrock's independence assumption (AKIA):} \emph{Packets maintain their lengths as they pass through the network and the service time (i.e., transmission time) at each link is directly
proportional to the packet length and channel capacity.}

\noindent However, the end-to-end delay modeling
problem under this assumption is considered to be analytically intractable \cite{Kleinrock2002}.



A PSN under AKIA can be modeled as a queueing network with
dependent service times. Two-node tandem networks with equal service times at both nodes were studied in \cite{Mitchell1977,Boxma1979,Boxma1979a,Pinedo1982}.
In 1977, Mitchell et al. \cite{Mitchell1977} investigated the effect of
service time correlation based on simulation experiments. Boxma \cite{Boxma1979}, \cite{Boxma1979a} derived an exact solution for
the stationary joint distributions of waiting times at both queues in $M/G/1 \rightarrow G/1$ system. Pinedo and Wolff \cite{Pinedo1982} compared expected waiting times under AKIA and KIA.
Two-node tandem networks when service times are proportional to message lengths were studied in \cite{Calo1981a}.

A few studies have reported tandem networks under AKIA with more than two queues. Rubin \cite{Rubin1975,Rubin1974,Rubin1976} considered a case of fixed packet length and approximated end-to-end delay formulae in a communication path. Expected waiting times
in a series of queues were investigated in \cite{Wolff1981}, \cite{Sandmann2010}.
In 2010, Chen et al. derived a new distribution function to model
end-to-end delay in a wireless multi-hop path \cite{Chen2010}. In 2011, Popescu and Constantinescu \cite{Popescu2011} reported results on network latency distribution of an actual chain of IP
routers experiments. By comparing the experimental results with predictions under KIA, they concluded that KIA isn't valid to model end-to-end delay distribution in these experiments.


A complex PSN with multiple flows can also be modeled as a multicommodity flow network in the field of the network flow theory \cite{Ahuja1993}. However, the theory assumes static flows so it is not for the case of stochastic flows that has a fluctuating nature.

For the past six decades, it is shown in the literature that analytical results on end-to-end delay under AKIA are limited to tandem networks. Moreover, only a two-node tandem network with equal capacity channels is known to have a non-explicit solution \cite{Boxma1979}, \cite{Boxma1979a}, but this solution isn't for the end-to-end delay. Kleinrock's question when PSNs have AKIA is still unanswered.


In this paper, I will answer Kleinrock's question on the end-to-end delay approximation in complex
PSNs under AKIA. Specifically, I first identify that the distribution derived in \cite{Chen2010} is a new special acyclic continuous-time phase-type distribution, which is
denoted $C(\mathbf{p}, \boldsymbol \theta)$. In the second step, I model a complex PSN with traffic flows as a multicommodity flow network in the field of the network flow theory; and then model a flow passing through such a network as a series of queues in the field of the queueing theory. In the last step, based on the results from the network flow theory and queueing theory, I apply $C(\mathbf{p}, \boldsymbol \theta)$ to develop a generalized method to approximate the end-to-end
delay (such as average end-to-end delay, jitter and density functions) for any traffic flow within the net under AKIA and KIA. Moreover, all the formulae related to this general method have simple and explicit forms.
In the simulation part, I consider two types of nets under AKIA: the first net is a tandem network with one traffic flow; the second net is a 100-node hypothetical network with 500 host pairs and 1000 traffic flows. By analyzing large simulation data sets from simulating the above two topologies, I show that the delay approximation method under AKIA provides close and much better approximation results than the method under KIA does when PSNs has AKIA.


The rest of this paper is organized as follows: the modeling procedure of a complex PSN with flows is described in Section II, i.e., a PSN with flows can be modeled as a multicommodity flow network; and each node in the net can be modeled as a queueing model. Section III is about developing a generalized approximation of end-to-end delay under AKIA and KIA by using $C(\mathbf{p}, \boldsymbol \theta)$, results from the multicommodity flow network and queueing theory. In Section IV, I will use the above methods under AKIA and KIA to approximate end-to-end delay in two topologically different network examples.
I show the approximation accuracy of the method under AKIA by comparing their approximation results with simulation results when PSNs have AKIA.
Section V summarizes this paper's work. The important
symbols and their definitions are listed in Table \ref{tab:not} for convenience.

\begin{table}[t]
\caption{Notations and Definitions}
\begin{tabular}{p{0.3in} p{2.9in}}
$\gamma$& average number of packets per second (or external arrival rate) of a specific flow \\
$\lambda(u,v)$& average number of packets (or total arrival rate) entering the channel per second on the channel $(u, v)$ \\
$\rho(u,v)$& traffic load of the channel $(u, v)$ \\
$1/\mu$& average length of packet \\
$\mu_X$& expected value of a random variable $X$, i.e., $E[X]$ \\
$\tilde{\mu}_X$& approximation of the mean value of a random variable $X$ \\
$\sigma_X$& the standard deviation of a random variable $X$ \\
$\tilde{\sigma}_X$& approximation of the standard deviation value of a random variable $X$ \\
$1/\mu$& average length of packet \\
$A_n$ & interarrival time of the $n^{\rm th}$ and the $n^{\rm n+1}$st packets of a specific flow \\
$c(u,v)$& channel capacity of the link $(u, v)$\\
$\overline{c}$& compound channel capacity of a specific flow along the path $p$, i.e., $\overline{c} = 1/\left(\sum_{j=i}^{h} 1/c({u,v})\right)$\\
$K_i$& the $i^{th}$ traffic flow, which constrains the information of source, sink and packet arrival rate \\
$p$& routing path of a specific flow\\
$R$& compound service time along the path $p$\\
$R(u,v)$& transmission time for the $n^{\rm th}$ packet passing through the link $(u,v)$\\
$\ddot{R}(u,v)$& transmission time for the $n^{\rm th}$ packet passing through the link $(u,v)$ when a packet length is chose under KIA${^{a}}$\\
$S_X(t)$& ccdf of a random variable $X$\\
$\tilde{S}_X(t)$& approximated ccdf of a random variable $X$\\
$s$& source node of a flow $i$ \\
$T(u,v)$& delay for the $n^{\rm th}$ packet passing through the link $(u,v)$ (includes both time on queue and time in transmission) under AKIA\\
$\ddot{T}(u,v)$& delay for the $n^{\rm th}$ packet passing through the link $(u,v)$ (includes both time on queue and time in transmission) under KIA\\
$d$ & destination node of a flow $i$ \\
$V$& length of the $n^{\rm th}$ packet under AKIA\\
$V(u,v)$& length of the $n^{\rm th}$ packet received node $u$ under KIA\\
$W$& compound queueing delay along the path $p$\\
$W(u,v)$& waiting time on the queue for the $n^{\rm th}$ packet passing through the link $(u,v)$\\
$Z$& end-to-end delay of the $n^{\rm th}$ packet of a specific flow under AKIA\\
$\ddot{Z}$& end-to-end delay of the $n^{\rm th}$ packet of a specific flow under KIA\\
\end{tabular}
\label{tab:not}
\footnotesize{ ${}^a$symbols with two dots hovered refer to quantities when KIA is made; symbols without these two dots refer to quantities under AKIA.}
\end{table}






\section{Packet-Switched Network and Stochastic Traffic Flow Modeling}

To model the end-to-end delay performance of stochastic traffic flows in PSNs,
let us first
model the network and flows as a multicommodity flow network model from a macroscopic viewpoint (Section \ref{sec:multi_flow_network}), and then model the stochastic traffic flow and queues
within a node as a queueing model from a microscopic viewpoint (section \ref{sec:traffic_model}).

\subsection{Multicommodity Flow Network Model}
\label{sec:multi_flow_network}

\begin{figure}[t]
\centering
\includegraphics[width=2.5in]{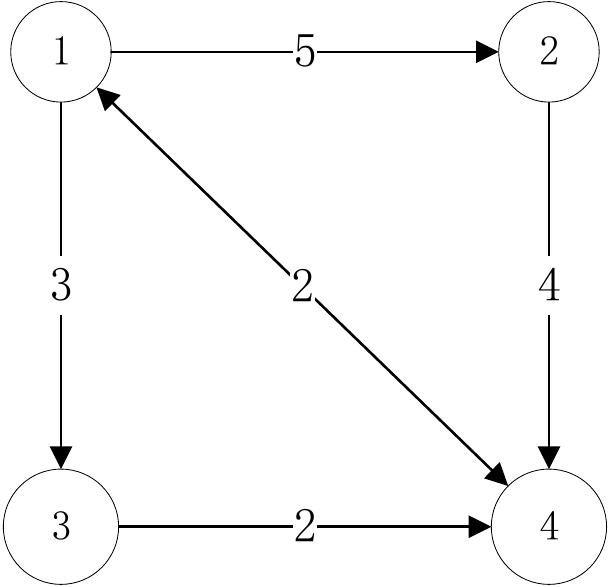}
\caption{Four-node network example.}
\label{fig:4_node sys}
\end{figure}

A PSN has some nodes (e.g., switches
and routers) and communication links. The topology of a PSN can be described as a directed graph $G\{V, E\}$, where
$V=\{1, 2, \cdots, N\}$ is a set of nodes and $E$ is a subset of the ordered pairs (arc) $(u, v)$ of elements taken from $V$. We call $(u, v)$ a one-way communication link (or link for short) and $c(u, v)$ is the channel capacity of the link $(u, v)$. For example, the network shown in Fig. \ref{fig:4_node sys} consists of
four nodes ($V=\{1, 2, 3, 4\}$ and six links $E=\{(1, 2), (1, 3), (2, 4), (3, 4), (1, 4), (4, 1)\}$. The symbol
``\begin{tikzpicture}[scale=2] \draw[>=triangle 90, <->](0,0.1) -- (0.3,0.1);\end{tikzpicture}'' in Fig. \ref{fig:4_node sys} indicates a full-duplex link.

There are $k$ traffic flows $K=\{K_{1},K_{2},\dots ,K_{k}\}$. For a traffic flow $i~(i\le k)$, we have $K_{i}=(s_{i},d_{i},\gamma_{i})$, where $s_{i} \in V$ and $d_{i} \in V$ are the source and the destination nodes of such a traffic flow, and $\gamma_{i}$ is the packet arrival rate in packets per unit time.
Let
$p_i = s_i x_2 x_3 \cdots x_{h} d_i$ (or $p_i = x_1 x_2 x_3 \cdots x_{h} x_{h+1}$ by letting $s_i = x_0$ and $d_i = x_{h+1}$) be a sequence of distinct nodes having the property
that $(s_i, x_2)$, $(x_j, x_{j+1})$ (for each
$j=2,\cdots h $) and $(x_{h}, d_i)$ are arcs (links) in $E$. We call such a sequence of nodes a path from $s_i$ to $d_i$. Moreover, the variable $f_i(u,v)$ defines the fraction of flow $i$ along the link $(u,v)$.

A fixed routing strategy is
considered, indicating that there exists only one path for any traffic flow, i.e., $f_i(u,v)\in \{0,1\}$ and the path is simple, i.e., no two nodes in $p_i~(i = 1, 2, \cdots, k)$ are identical. Table \ref{tab:rout_tab} shows a central routing directory \cite{Stallings2013} (essentially a $4\times 4$ routing table matrix $R$) for the network example of Fig. \ref{fig:4_node sys}. When packets of a flow enter a node, this directory identifies
the next node to visit on the route based on the flow's destination node. When $G(V,E)$, $K$ and $R$ are known,
we have the total flow arrival rate on the link $(u,v)$ (this quantity plays a central role in this work):
\begin{equation}
\lambda(u,v) = \sum^{k}_{i=1} f_i(u,v)\gamma_i.
\label{eq:total_arri}
\end{equation}
Let $1/\mu$ be the average packet length of flows. To ensure network stability, the total flow data rate on any link shouldn't exceed their capacity:
\begin{equation}\label{x}
\frac{\lambda(u,v)}{\mu} \leq c(u,v),\forall (u,v)\in E.
\end{equation}


\begin{table}[t]
\centering
\caption{Central routing directory example for the four-node example in Fig. \ref{fig:4_node sys}.}
\begin{tabular}{|l|*{4}{c|}}\hline
\backslashbox{To Node}{Current Node}
&\makebox[1em]{1}&\makebox[1em]{2}&\makebox[1em]{3}
&\makebox[1em]{4}\\\hline
1 &--&2&3&4\\\hline
2 &4&--&4&4\\\hline
3 &4&4&--&4\\\hline
4 &1&1&1&--\\\hline
\end{tabular}
\label{tab:rout_tab}
\end{table}

\subsection{Queueing Model for Stochastic Traffic Model and Store-and-Forward Mechanism}
\label{sec:traffic_model}

For flow $i$ (for simplicity, the flow index $i$ is usually dropped if no ambiguity is raised), it has the following two properties:
\begin{enumerate}
	\item Packet arrivals follow a Poisson process with the mean arrival rate $\gamma$;
	\item Packet lengths are independent exponentially distributed random variables with the
mean packet length $1/\mu$.
\end{enumerate}
For example, let $A_n$ and $V_n$ be the time between the arrivals of the $n^{\rm th}$ and the $(n+1)^{\rm st}$ packets of this flow, and the $n^{\rm th}$ packet length, respectively. The pdf of $A_n$ is given by
\begin{equation}
f_{A_n}(x) = \gamma {e^{ - \gamma x}}.
\label{eq:2}
\end{equation}
The pdf of packet length $V_n$ is given by
\begin{equation}\label{message_length}
f_{V_n}(x) = \mu {e^{ - \mu x}}.
\end{equation}
Moreover, KIA and AKIA both imply that 1) the channel capacities of links are deterministic and 2) the retransmission rate of any link is small and can be ignored.
Therefore, when the $n^{\rm th}$ packet that passes through a link $(u,v)$, the service time $R_n(u,v)$ of this packet is
\begin{equation}\label{eq:service_time_def}
  R_n(u,v) = \frac{V_n}{c(u,v)}.
\end{equation}

Any node in a PSN implements a store-and-forward mechanism. When packets are passing through a node, they are stored in a queue, if necessary, and then forwarded (transmitted) to
the next node based on a routing strategy. The reference design architecture of routers is shown in Fig. \ref{fig:router_arch}. Such a mechanism together with the above stochastic traffic model can be sufficiently characterized by a queueing model. Finally, the arrival rate $\gamma$ and
the average packet length $1/\mu$ both appear in the multicommodity flow model in Section
\ref{sec:multi_flow_network} and the queueing model in Section \ref{sec:traffic_model};
and they bridge these two models from a macroscopic view to a microscopic view.


\begin{figure}[t]
\centering
\includegraphics[width=3.5in]{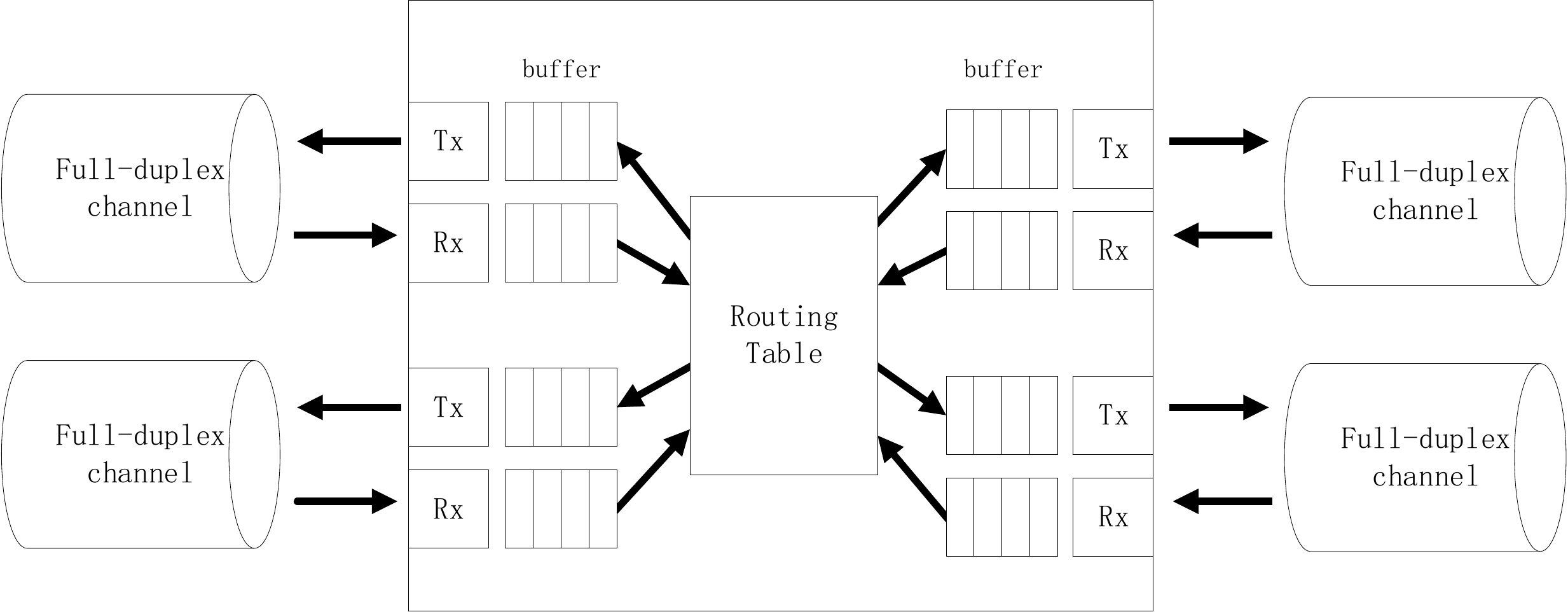}
\caption{Reference router architecture.}
\label{fig:router_arch}
\end{figure}

\section{New Phase-Type Distribution $C(\mathbf{p}, \boldsymbol \theta)$ and End-to-End Delay Approximation under AKIA and KIA}
\label{sec:packet}

Given that we have the complete information of $G(V,E)$, $K$, $R$, $c(u,v)$ and $\lambda(u,v)$, let us
formulate a framework for end-to-end delay approximation of any flow in PSNs.
We first 1) focus on any flow $i$ ($K_i=(s_i,d_i,\gamma_i)$) with source $s$ and destination $d$ (the index $i$ is again dropped for simplicity) and 2) the $n^{\rm th}$ packet of this flow with its length being $V$ (for the same simplicity purpose, the index $n$ is dropped). Under fixed routing, this packet shall follow a path $p=x_1x_2x_3\cdots x_h x_{h+1}$ (by letting $s = x_1$ and $d = x_{h+1}$).

%

Denote by $T(x_j,x_{j+1})$ the delay of the $n^{\rm th}$ packet
experienced when it passes through the link $(x_j, x_{j+1})$ along the path $p$. $T(x_j,x_{j+1})$ is the sum of the queueing delay $W(x_j,x_{j+1})$ and
the service time $R(x_j,x_{j+1})$ (the processing and propagation delays are neglected):
\begin{equation}\label{eq:tn_wn_sn}
T(x_j,x_{j+1}) = W(x_j,x_{j+1}) + R(x_j,x_{j+1}).
\end{equation}
According to \eqref{eq:service_time_def}, the service time of the $n^{\rm th}$ packet at the link $(x_j,x_{j+1})$ is
\begin{equation}\label{r_n_sk}
  R(x_j,x_{j+1}) = \frac{V}{c(x_j,x_{j+1})}.
\end{equation}
The end-to-end packet delay $Z$ of the $n^{th}$ packet from source $s$ to destination $d$ along the path $p$ is the sum of delays
$T(x_j,x_{j+1})~(j=\{1,2,\cdots, h\})$:
\begin{align}\label{znij}\nonumber
Z \mathop  = \limits^{(1)}  &\sum_{j=1}^{h} {T(x_j,x_{j+1})} \mathop  = \limits^{(2)}  \sum_{j=1}^{h}  {\left(W(x_j,x_{j+1}) + R(x_j,x_{j+1})\right)} \\
\mathop  =  \limits^{(3)}& \sum_{j=1}^{h}  {\left(W(x_j,x_{j+1}) + \frac{V}{c(x_j,x_{j+1})}\right)},
\end{align}
where the $2^{\rm nd}$ equality holds because of \eqref{eq:tn_wn_sn} and the $3^{\rm rd}$ equality holds because of \eqref{r_n_sk}. Eq. \eqref{znij} breaks down an end-to-end packet delay into smaller delays.



Let's then introduce three assumptions:
\begin{description}
  \item[A1:] The flow into each queue of the network is Poisson;
  \item[A2:] The queueing delays $W(x_j,x_{j+1}) (j=\{0,1,2,\cdots, h\})$ at each link are mutually independent;
  \item[A3:] The service times and queueing delays are mutually
independent.
\end{description}
These three assumptions must be made because they allow us to analyze queues in isolation so the end-to-end delay approximation problem in a complex network can be simplified to an approximation problem in a series of queues. However, these assumptions are in general not true.
Consider a two-node tandem network under AKIA.
The fact that the interarrival times at the second queue are strongly correlated with the packet lengths (see p. 210 in \cite{Bertsekas1992}) rejects the assumptions A1-A3. Further consider a PSN with KIA, which is essentially an open multiclass Jackson network. A path
has to satisfy the non-overtaking condition to ensure independence \cite{Walrand1980}.

The last ingredient to approximating the delay is a new
distribution function $C(\mathbf{p}, \boldsymbol \theta)$, which will be first introduced in
Section \ref{sec:on_the_dist}.
Eq. \eqref{znij} together with the assumptions A1-A3 and $C(\mathbf{p}, \boldsymbol \theta)$ serve as the basis for solving the end-to-end delay approximating problem under AKIA in Section \ref{sec:wo_kia} and under KIA in Section \ref{sec:w_kia}.


%

\subsection{On the Distribution in \cite{Chen2010}}
\label{sec:on_the_dist}

We now briefly present the distribution in \cite{Chen2010}. Let $h\geq 1$ be a positive integer, $p_1,p_2,\cdots,p_h$ be
$h$ real numbers between 0 and 1, and
$X_1,X_2,\cdots,X_h$ be $h$ independent
Bernoulli distributed random variables (rvs) with parameter $p_1, p_2, \cdots, p_k$; the probability
mass function (pmf) of $X_f$ is given by
\begin{equation}
{\displaystyle f_{X_f}(t)={\begin{cases}p_f&{\text{if }}t=1,\\q_f=1-p_f&{\text{if }}t=0.\end{cases}}}
\label{eq:1}
\end{equation}
Let ${\theta}_1,{\theta}_2,\cdots,{\theta}_h$ be
$h$ positive real numbers and  $Y_1, Y_2 ,\cdots, Y_h$ be $h$ independent
rvs, with $Y_f$ following an exponential distribution of parameter $\theta_f$; the probability density function (pdf) of $Y_f$ is given by
\begin{equation}
{f_{{Y_f}}}(t) = {\lambda_f}{e^{ - {\theta_f}t}}.
\label{eq:1}
\end{equation}
If $U_f$ is the product of $X_f$ and $Y_f$ (i.e., $U_f=X_f Y_f$), then $U_1, U_2 ,\cdots, U_h$ are $h$ independent
rvs; the pdf of $U_f$ is \cite{Chen2010}
\begin{equation}
{f_{{U_f}}}(t) = {p_f}{\theta_f}{e^{ - {\theta_f}t}} + {q_f}\delta (t).
\label{eq:1}
\end{equation}
The cumulative density function (cdf) of $U_i$ is
\begin{equation}
P(U_f \le t) = 1-{p_f}{e^{ - {\theta_f}t}}.
\label{eq:v_i_pdf}
\end{equation}

The random variable $Z=\sum_{f=1}^h U_{f}$ follows a distribution in \cite{Chen2010} with two parameter sets ${\mathbf{p}}=\{p_1,p_2,\cdots,p_h\}$ and
${\boldsymbol{\theta}=\{\theta{}}_1,{\theta{}}_2,\cdots,{\theta{}}_h\}$. Under the condition of ${\theta{}}_f\neq{\theta{}}_g~(f\neq g, \forall f, g \le h)$, the pdf of $Z$ is given by
\begin{equation}
{f_Z}(t) = \sum\limits_{f = 1}^h {\frac{{{\theta_f}}}{{{P_f}}}{e^{ - {\theta
_f}t}}}  + Q\delta \left( t \right),
\label{eq:ph_pdf}
\end{equation}
where
\begin{equation}\label{eq:p_i}
\frac{1}{{{P_f}}} = {p_f}\prod\limits_{g = 1, g \ne f}^h \left({1+\frac{{{p_g}{\theta_f}}}{{{\theta_g} - {\theta_f}}}}\right)
\end{equation}
and
\begin{equation}
Q = 1 - F(0) = \prod\limits_{f = 1}^h {{q_f}}.
\label{eq:6}
\end{equation}
The complementary cumulative density function (ccdf) of
$Z$ is given by \cite{Chen2010}
\begin{equation}
P(Z > t) = {S}_Z(t) =  \sum\limits_{f = 1}^h {\frac{{1}}{{{P_f}}}{e^{ - {\theta
_f}t}}}.
\label{eq:ph_cdf-1}
\end{equation}
Moreover, the mean and the variance of $Z$ are \cite{Chen2010}
\begin{equation}
\mu_Z = E[Z] = \sum\limits_{f = 1}^h {\frac{{{p_f}}}{{\theta_f}}},
\label{eq:12-mean}
\end{equation}
and
\begin{equation}\label{eq:14-1-variance}
  \sigma^2_Z = Var[Z] = \sum\limits_{f = 1}^h {\left(
{\frac{{2{p_f}}}{{\lambda _f^2}} - {{\left( {\frac{{{p_f}}}{{{\theta_f}}}}
\right)}^2}} \right)}.
\end{equation}
The distribution with a pdf of \eqref{eq:ph_pdf} or a ccdf of \eqref{eq:ph_cdf-1} is an acyclic continuous-time phase-type distribution; it is denoted as $C(\mathbf{p},\boldsymbol \theta)$, i.e.,
$Z \sim C(\mathbf{p},\boldsymbol \theta)$.
Its Markov chain and phase-type related properties
are discussed in Appendix \ref{sec:prop} in detail. Note that these properties
will later be used in this work.




In the next part, I will show how to use this $C(\mathbf{p}, \boldsymbol \theta)$ to approximate end-to-end delay
in PSNs under AKIA as well as KIA.

\subsection{Delay Approximation Method for PSNs under AKIA Using $C(\mathbf{p}, \boldsymbol \theta)$}
\label{sec:wo_kia}

Eq. \eqref{znij} can be expanded as
\begin{equation}\label{znij_wo-1}
Z = \sum_{j=1}^{h} {W(x_j,x_{j+1})}  + \sum_{j=1}^{h} \frac{V}{c(x_j,x_{j+1})} = W + R,
\end{equation}
where $W$ and $R$ are called the compound queueing delay and the compound service time along the path $p$ in this paper.

Because packets keep their lengths unchanged, the compound service time in \eqref{znij_wo-1} is
\begin{equation}\label{rnij_wo}
R = \sum_{j=1}^{h} \frac{V}{c(x_j,x_{j+1})} = V \sum_{j=1}^{h} \frac{1}{c(x_j,x_{j+1})}= \frac{V}{\overline{c}},
\end{equation}
where
\begin{equation}\label{eq:c_ij_sk}
  \frac{1}{\overline{c}} = \sum_{j=1}^{h} \frac{1}{c(x_j,x_{j+1})}.
\end{equation}
$\overline{c}$ is termed compound channel capacity along the path $p$.
Because the packet length is exponential with a pdf of \eqref{message_length}, \eqref{rnij_wo} indicates that $R$ is
another exponentially distributed rv with the mean compound service time being
\begin{equation}\label{a}
E[R] = E[V] \frac{1}{\overline{c}} = \frac{1}{\mu \overline{c}} = \frac{1}{\theta_R}.
\end{equation}
So the pdf of the compound service time $R$ is
\begin{equation}
f_{R}(t) = \mu \overline{c} {e^{ - \mu \overline{c} t}} = \theta_R {e^{ - \theta_R t}}.
\label{eq:compound_exp}
\end{equation}
Remark \ref{pro:deg_exp_erl} in Appendix \ref{sec:prop} shows that an exponential distribution is a special case of $C(\mathbf{p}, \boldsymbol \theta)$. Hence, we have $R \sim C(1, \theta_R)$.



If the assumption A1 is made, then a path $p$ can be decomposed into
$h$ individual $M/M/1$ queues (for each link $(x_j, x_{j+1})$ along the path). The flow into each link is a Poisson process with mean rate ${\lambda{}}({x_j,x_{j+1}})$ from \eqref{eq:total_arri}.
Moreover, it is a property of an $M/M/1$ queue that the stationary distribution function of the
queueing delay $W(x_j,x_{j+1})$ at the link $(x_j, x_{j+1})$ is

\begin{equation}\label{eq:w_sk_cdf}
\begin{array}{l}
P(W(x_j,x_{j+1}) \le t)
\approx 1 - {\rho (x_j,x_{j+1})} e^{{ - {\theta (x_j,x_{j+1})}t}},
\end{array}
\end{equation}
where
\begin{equation}\label{rhosk}
\rho(x_j,x_{j+1}) = \frac{{\lambda{}}(x_j,x_{j+1})}{\mu c(x_j,x_{j+1})}
\end{equation}
and
\begin{equation}\label{theta_sk}
\theta (x_j,x_{j+1}) = \mu c(x_j,x_{j+1}) - {\lambda{}}(x_j,x_{j+1}).
\end{equation}
By comparing the cdf of \eqref{eq:v_i_pdf} with the cdf of \eqref{eq:w_sk_cdf}, we can confirm that the rv $W(x_j,x_{j+1}) \sim C(\rho (x_j,x_{j+1}), \theta (x_j,x_{j+1}))$ approximately.

Based on the assumption A2, the compound queueing delay $W$ is
the sum of $W(x_j,x_{j+1})$ along the path so $W$ approximately follows $C(\mathbf{p}_W, \boldsymbol{\theta}_W)$ with
\begin{equation}\label{eq:p_w}
\mathbf{p}_{W} = \{\rho(x_1,x_2), \cdots, \rho(x_j,x_{j+1}), \cdots,\rho(x_h,x_{h+1})\}
\end{equation}
and
\begin{equation}\label{eq:theta_w}
\boldsymbol{\theta}_W = \{\theta(x_1,x_2), \cdots, \theta(x_j,x_{j+1}), \cdots, \theta(x_h,x_{h+1})\}.
\end{equation}
%
Speaking intuitively, consider a patient who must perform a number of tests in a hospital and each test has a queue. There is a chance that the patient does not have to wait in a queue if he/she finds the queue is empty. The total time the patient spent in all the queues is a $C(\mathbf{p}_W, \boldsymbol \theta_W)$ distributed rv.


Eq. \eqref{znij_wo-1} shows that the end-to-end delay $Z$ is
the sum of the compound queueing delay $W$ ($W \sim C(\mathbf{p}_W, \boldsymbol \theta_W)$ approximately) and the compound service time $R$ ($R \sim C(1, \theta_R)$). By using the assumption A3, I have the
following result:

\begin{prop}
\label{prop:delay_akia}
In a PSN under AKIA, consider a flow with a fixed path $p$. The end-to-end delay
$Z$ can be described by $C(\mathbf{p}_Z, \boldsymbol \theta_Z)$ approximately with
\begin{equation}\label{eq:p_w_1}
\mathbf{p}_{Z} = \mathbf{p}_{W}\cup \{1\} = \{\rho(x_1,x_2), \rho(x_2,x_3), \cdots, \rho(x_h,x_{h+1}), 1\}
\end{equation}
and
\begin{equation}\label{eq:theta_w_r}
\boldsymbol{\theta}_Z = \boldsymbol{\theta}_{W}\cup \{\theta_R\} = \{\theta(x_1,x_2), \theta(x_2,x_3), \cdots, \theta(x_h,x_{h+1}, \theta_R)\}.
\end{equation}
\end{prop}

Based on Proposition \ref{prop:delay_akia} and , we have the approximated pdf and ccdf of delay $Z$ (according to \eqref{eq:ph_pdf} and \eqref{eq:ph_cdf-1}):
\begin{equation}
{f}_{Z}(t) \approx {\tilde{f}}_{Z}(t) = \sum\limits_{f = 1}^{h+1}
\frac{\theta_f}{P_f} e^{ - {\theta
_f}t}
\label{eq:ph_pdf_ud_akia}
\end{equation}
and
\begin{equation}
{S}_{Z}(t) \approx {\tilde{S}}_{Z}(t) = \sum\limits_{f = 1}^{h+1} \frac{1}{P_f} e^{ - {\theta
_f}t},
\label{eq:ph_cdf_delay}
\end{equation}
where
\begin{equation}
\frac{1}{P_f} = {p_f}{\prod\limits_{\substack{g = 1,g \ne f}}^{h+1} \left({1+\frac{{{p_g}{\theta _f}}}{{{\theta_g} - {\theta _f}}}}\right){}},
~~ p_{f/g} \in \mathbf{p}_Z, \theta_{f/g} \in \boldsymbol{\theta}_Z,
\label{eq:pf_ud_akia}
\end{equation}
Moreover, the approximated average delay and the jitter are
\begin{equation}\label{eq:wokia_mean-rm}
E[Z]\approx  \tilde{\mu}_Z = \sum\limits_{j = 1}^h E[T(x_j,x_{j+1})]
= \sum\limits_{f = 1}^h \frac{1}{\theta_{f}}, ~~\theta_f \in \boldsymbol \theta_W
\end{equation}
and
\begin{equation}\label{eq:appr_jitter_ud_akia}
  \sigma_{Z}\approx \tilde{\sigma}_{Z}=\sqrt{Var[Z]} = \sqrt{\sum\limits_{f = 1}^{h+1} {\left(
{\frac{{2{p_f}}}{{\theta_f^2}} - {{\left( {\frac{{{p_f}}}{{{\theta_f}}}}
\right)}^2}} \right)}},~p_{f} \in \mathbf{p}_Z, \theta_f \in \boldsymbol{\theta}_Z.
\end{equation}
For the result of \eqref{eq:wokia_mean-rm}, see Appendix \ref{sec:proof_avg}. Eq. \eqref{eq:appr_jitter_ud_akia}
is trivial from \eqref{eq:14-1-variance}.

%




%

\subsection{Delay Approximation Method for PSNs under KIA}
\label{sec:w_kia}

We can roughly follow the same steps in Section \ref{sec:wo_kia} to develop the end-to-end delay approximation in a PSN under KIA.
Consider again a flow $i$ with source $s$ and destination $d$ and the $n^{\rm th}$ packet that goes through a sequence of nodes $\{x_1, x_2, \cdots, x_{h+1}\}$.
With KIA, each time the packet is received at a node (say node $j$, in other words, at a link
$(x_j,x_{j+1})$) within the net, a
new length, $V(x_j,x_{j+1})$, is chosen for this packet from the pdf of \eqref{message_length}. Therefore, $V(1,2), \cdots, V(x_j,x_{j+1}), \cdots, V(h,h+1)$ are independent and identically distributed rvs identical to the rv $V$.

Under KIA, \eqref{znij} can be written as
\begin{align}\nonumber
\ddot{Z} &= \sum_{j=1}^{h} {\ddot{T}(x_j,x_{j+1})} = \sum_{j=1}^{h} {W(x_j,x_{j+1})}  + \sum_{j=1}^{h} \frac{V(x_j,x_{j+1})}{c(x_j,x_{j+1})}\\\nonumber\label{znij_wo}
&=\sum_{j=1}^{h}\left( {W(x_j,x_{j+1})}  + \frac{V(x_j,x_{j+1})}{c(x_j,x_{j+1})}\right)\\
&=\sum_{j=1}^{h}\left( {W(x_j,x_{j+1})}  + \ddot{R}(x_j,x_{j+1})\right),
\end{align}
where $\ddot{R}(x_j,x_{j+1})$ is the service time at link $(x_j,x_{j+1})$.
If the assumption A1 is made, the stationary distribution function of the
queueing delay $W(x_j,x_{j+1})$ at the link $(x_j, x_{j+1})$ is
\begin{equation}
\begin{array}{l}
P(W(x_j,x_{j+1}) \le t)
\approx 1 - {\rho (x_j,x_{j+1})} e^{{ - {\theta (x_j,x_{j+1})}t}},
\end{array}
\label{eq:10}
\end{equation}
where $\rho (x_j,x_{j+1})$ and $\theta (x_j,x_{j+1})$ are defined in \eqref{rhosk} and \eqref{theta_sk}, respectively.
Note that $W(x_j,x_{j+1})$ is again a $C(\rho (x_j,x_{j+1}), \theta (x_j,x_{j+1}))$ distributed rv approximately.

According to \eqref{znij_wo}, $\ddot{T}(x_j,x_{j+1})$ is the sum of $W(x_j,x_{j+1})$ and $\ddot{R}(x_j,x_{j+1})$. Considering that the service time $\ddot{R}(x_j,x_{j+1})$ is exponential with the mean service time of $1/\mu c(x_j,x_{j+1})$, I have the following result:
\begin{remark}
\label{rem:exp}
$\ddot{T}(x_j,x_{j+1})$ is exponential with the mean delay of $1/\theta (x_j,x_{j+1})$; the cdf of $\ddot{T}(x_j,x_{j+1})$ is given by
\begin{equation}
P(\ddot{T}(x_j,x_{j+1}) \le t)
\approx 1 -  e^{{ - {\theta (x_j,x_{j+1})}t}}.
\label{eq:10}
\end{equation}
\end{remark}
For a proof of Remark \ref{rem:exp}, see Appendix \ref{sec:proof_exp}.
The end-to-end delay $\ddot{Z}$ is the sum of the delay $\ddot{T}(x_j,x_{j+1})$ at all links along the path. Because $\ddot{T}(x_j,x_{j+1})$ is approximately exponential, the stationary distribution of $\ddot{Z}$ is approximately hypoexponential, i.e., $\ddot{Z}\sim Hypo(\boldsymbol\theta_W)$ approximately. Most importantly, Remark \ref{pro:deg_exp_erl} in Appendix \ref{sec:prop} shows that a hypoexponental distribution is a special case of $C(\mathbf{p}, \boldsymbol \theta)$ as well. Now we have a result for PSNs under KIA:

\begin{prop}
\label{prop:delay_kia}
In a PSN under KIA, consider a flow with a fixed path $p$.
$\ddot{Z}$ can be described by $C(\mathbf{1}, \boldsymbol \theta_W)$ approximately, where $\mathbf{1}$ is a set of ``1''s and $\boldsymbol \theta_W$ is defined in \eqref{eq:theta_w}.
\end{prop}

Based on Proposition \ref{prop:delay_kia}, the approximated pdf and ccdf of $\ddot{Z}$ are
\begin{equation}\label{eq:pdf_w_kia}
{f}_{\ddot{Z}}(t) \approx {\tilde{f}}_{\ddot{Z}}(t) = \sum\limits_{f = 1}^h {\prod\limits_{g = 1, g \ne f}^h \left({\frac{{\theta_{g}}}{{{\theta _{g}} - {\theta _{f}}}}}\right)}\theta_f
{e^{ - {\theta_f}t}},
~\theta_{f/g} \in \boldsymbol{\theta}_W,
\end{equation}
and
\begin{equation}\label{eq:ccdf_w_kia}
{S}_{\ddot{Z}}(t) \approx {\tilde{S}}_{\ddot{Z}}(t) = \sum\limits_{f = 1}^h {\prod\limits_{g = 1, g \ne f}^h \left({\frac{{\theta_{g}}}{{{\theta _{g}} - {\theta _{f}}}}}\right)}{e^{ - {\theta_f}t}},
~\theta_{f/g} \in \boldsymbol{\theta}_W,
\end{equation}
where $\boldsymbol{\theta_W}$ is defined in \eqref{eq:theta_w}. The distribution function of \eqref{eq:ccdf_w_kia} under KIA may first be found by Wong in 1978 \cite{Wong1976} and was also used in Popescu and Constantinescu's work \cite{Popescu2011}.
Moreover, the case $\theta_f = \theta_g$ in \eqref{eq:ph_pdf_ud_akia}, \eqref{eq:ph_cdf_delay},
\eqref{eq:pdf_w_kia} and \eqref{eq:ccdf_w_kia} is pathological as it causes a division by zero problem. Such a problem can be mitigated by slightly changing either the value of $\theta_f$ or the value of $\theta_g$.

The approximated average delay and the jitter for flow $i$ are
\begin{equation}
E[\ddot{Z}] \approx \hat{\mu}_{\ddot{Z}} = \sum\limits_{j = 1}^h E[{\ddot{T}}(x_j,x_{j+1})]=\sum\limits_{f = 1}^h {\frac{{{1}}}{{\theta _f}}}, ~~\theta_f \in \boldsymbol \theta_W
\label{eq:wokia_mean-1}
\end{equation}
and
\begin{equation}\label{eq:jitter_ud_kia}
  \sigma_{\ddot{Z}} \approx \tilde{\sigma}_{\ddot{Z}} = \sqrt{Var[\ddot{Z}]} = \sqrt{\sum\limits_{f = 1}^h {\frac{1}{\theta _f^2}}},~~\theta_f \in \boldsymbol{\theta}_W.
\end{equation}
\eqref{eq:wokia_mean-1} and \eqref{eq:jitter_ud_kia} are the basic properties of the hypoexponential distribution. We have a few remarks below that are useful for analyzing simulation results in the next section.



\begin{remark}\label{rm:one-hop}
When the path length is one (i.e., $h=1$), $Z = \ddot{Z}$.
\end{remark}

The proof of this remark is trivial: If the path length is one, then both the PSN under AKIA and the PSN under KIA can be modeled as an $M/M/1$ system with the same parameters,
indicating that $Z = \ddot{Z}$.

\begin{remark}\label{rm:same_avg}
The approximated average end-to-end delays under AKIA and under KIA are the same, i.e.,
\begin{equation}\label{eq:equal_mean}
\tilde{\mu}_Z = \tilde{\mu}_{\ddot{Z}}.
\end{equation}

\end{remark}
By comparing \eqref{eq:wokia_mean-rm} with \eqref{eq:wokia_mean-1}, we can conclude that \eqref{eq:equal_mean} or Remark \ref{rm:same_avg} is true.
Furthermore, either $E[T(x_j,x_{j+1})]$ of \eqref{eq:wokia_mean-rm} or $E[\ddot{T}(x_j,x_{j+1})]$ of \eqref{eq:wokia_mean-1} has the exact same form as (1) in Kleinrock's work \cite{Kleinrock2002}. Hence,
the average end-to-end delay derived by Kleinrock is still valid.


\begin{remark}\label{rm:less}
The approximated jitter under AKIA is higher than the approximated jitter under KIA, i.e.,
\begin{equation}\label{eq:std-wo-w}
  \tilde{\sigma}_{Z} > \tilde{\sigma}_{\ddot{Z}}.
\end{equation}

\end{remark}
For the proof of Remark \ref{rm:less}, see Appendix \ref{sec:less}. Finally, Propositions \ref{prop:delay_akia} and \ref{prop:delay_kia} show that the approximated end-to-end delay under AKIA and KIA are two different $C(\mathbf{p}, \boldsymbol \theta)$ distributed rvs, i.e., $Z \sim C(\mathbf{p}_Z, \boldsymbol \theta_Z)$ and $Z \sim C(\mathbf{1}, \boldsymbol \theta_W)$ approximately. Their distribution have simple forms of average delay, jitter and density functions.

\section{End-to-End Delay Results and Discussion}


In this section, I put the delay approximation methods under AKIA and KIA together and compare them
with simulation results from two topologically different networks under AKIA: a tandem network with one traffic flow (see Section \ref{sec:two-node}) and a 100-node hypothetical network with 1000 traffic flows (see Section \ref{sec:100-node}). Average end-to-end delay won't be considered because both methods produce the same approximation (see Remark \ref{rm:same_avg} in Section \ref{sec:w_kia}).



\subsection{Tandem Networks under AKIA}
\label{sec:two-node}

\begin{figure}[t]
\centering
\includegraphics[width=3.5in]{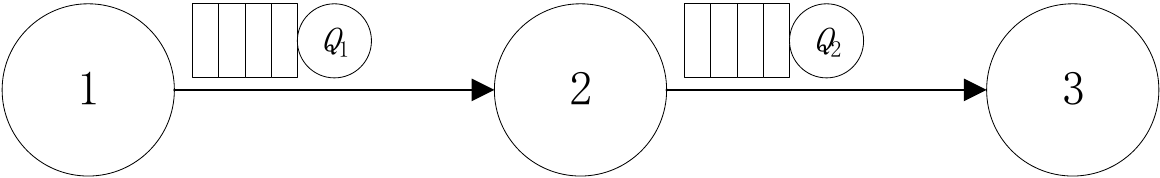}
\caption{Two-node tandem network.}
\label{fig:2-node}
\end{figure}

Let us start with a simplistic two-node tandem network under AKIA shown in Fig. \ref{fig:2-node}. The channel capacities of all links are the same, i.e., $c(1,2) = c(2,3) = c$. There is only one traffic flow $K_1 = \{s = 1, d = 3, \gamma\}$.
The total arrival rates at both queues are the same, i.e.,
\begin{equation}\label{eq:two_node_gamma}
\lambda(1,2) = \lambda(2,3) = \gamma.
\end{equation}
For demonstration purpose, let the average service times at both queues equal to one \cite{Pinedo1982}, i.e.,
\begin{equation}\label{eq:mu_c_1}
\mu c = 1.
\end{equation}

Such a network is essentially an $M/M/1 \rightarrow{} M/1$ system where any given packet has equal service times at two queues. Note that the same simulation setting has been done before by \cite{Mitchell1977,Boxma1979a,Pinedo1982} but is first reproduced for end-to-end delay performance. Moreover, let $\rho$ be the traffic load of the link $(1,2)$:
\begin{equation}\label{x}
\rho \mathop  = \limits^{(1)} \rho(1,2) \mathop  = \limits^{(2)} \frac{\lambda(1,2)}{\mu c} \mathop  = \limits^{(3)}  \frac{\gamma}{\mu c}  \mathop  = \limits^{(4)} \gamma.
\end{equation}
The $3^{\rm rd}$ equality holds because of \eqref{eq:two_node_gamma} and the $4^{\rm th}$ equality holds because of \eqref{eq:mu_c_1}. The values of $\rho$ that I considered
are 0.005, 0.05, 0.1, 0.2, 0.3, 0.4, 0.5, 0.6, 0.62, 0.6612, 0.7, 0.8, 0.9, 0.99. The number of packets simulated in each simulation run was 1 million (1,000,000).

Based on the methods I described in Sections \ref{sec:wo_kia} and \ref{sec:w_kia},
the end-to-end delay $Z$
under AKIA can be described by $C(\mathbf{p}_Z, \boldsymbol{\theta}_Z)$ approximately with
\begin{equation}\label{eq:pz_2h_ud_akia}
\mathbf{p}_Z = \{p_1 = \rho(1,2)=\rho, p_2 = \rho, 1\}
\end{equation}
and
\begin{equation}\label{eq:tz_2h_ud_akia}
\boldsymbol{\theta}_Z = \boldsymbol{\theta}_W \cup \theta_3,
\end{equation}
where
\begin{equation}\label{eq:tw_2h_ud_kia}
\boldsymbol{\theta}_W = \{\theta_1 = \mu c - \lambda = 1 - \rho, \theta_2 = 1 - \rho\}
\end{equation}
and
\begin{equation}\label{xx}
\theta_3 = \mu \overline{c} = \mu\frac{1}{\frac{1}{\overline{c}}}
\mathop  = \limits^{(1)} \mu\frac{1}{\frac{1}{c({1,2})}+\frac{1}{c({2,3})}} = \frac{1}{\frac{1}{\mu c({1,2})}+\frac{1}{\mu c({2,3})}} = \frac{1}{2}.
\end{equation}
The end-to-end delay $\ddot{Z}$ under KIA can be described by $C(\mathbf{1}, \boldsymbol{\theta}_W)$ (or equivalently $\ddot{Z} \sim Hypo(\boldsymbol{\theta}_W)$) approximately.

Furthermore, the Markov chains for $C(\mathbf{p}, \boldsymbol \theta)$
and $Hypo(\boldsymbol{\theta}_W)$ are shown in Fig. \ref{fig:markov_2hop}(a) and (b), respectively.

\begin{figure}[htp]
\centering

\subfloat[Markov chain for $C(\mathbf{p}_Z, \boldsymbol \theta_Z)$.]{%
  \includegraphics[clip,width=\columnwidth]{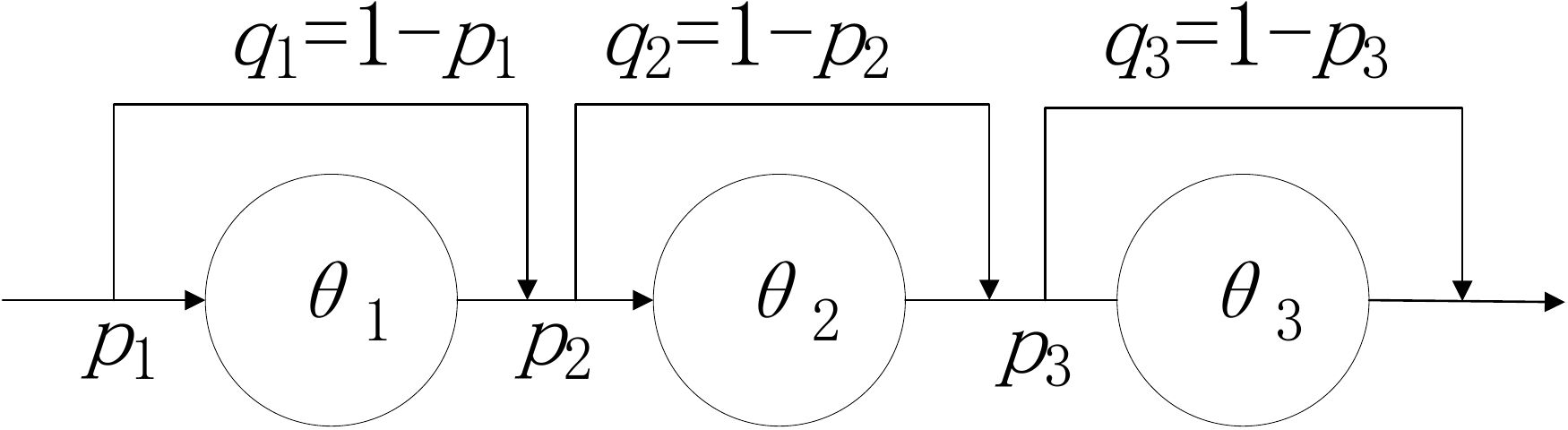}%
}

\subfloat[Markov chain for $Hypo(\boldsymbol \theta_W)$.]{%
  \includegraphics[clip,width=0.6\columnwidth]{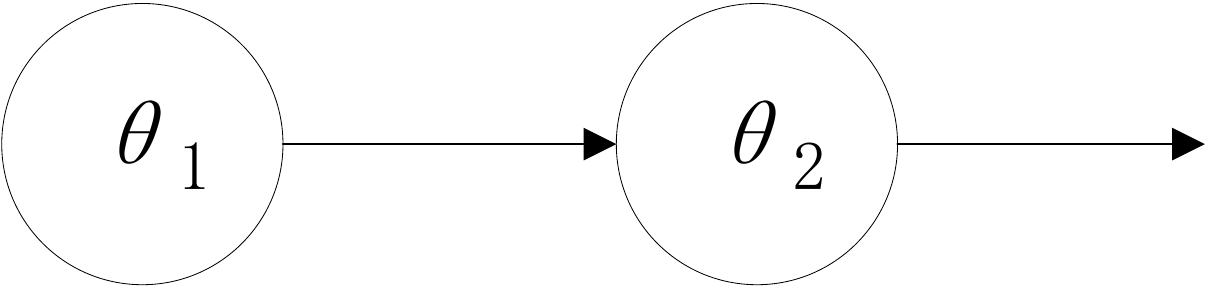}%
}

\caption{Markov chains for $C(\mathbf{p}_Z, \boldsymbol \theta_Z)$ and
$Hypo(\boldsymbol \theta_W)$ that approximate end-to-end delay distribution in a two-node tandem network
under AKIA and KIA, respectively.}
\label{fig:markov_2hop}
\end{figure}

\subsubsection{Jitter Approximation Comparison}

\begin{figure*}[t]
\centering
\includegraphics[width=7in]{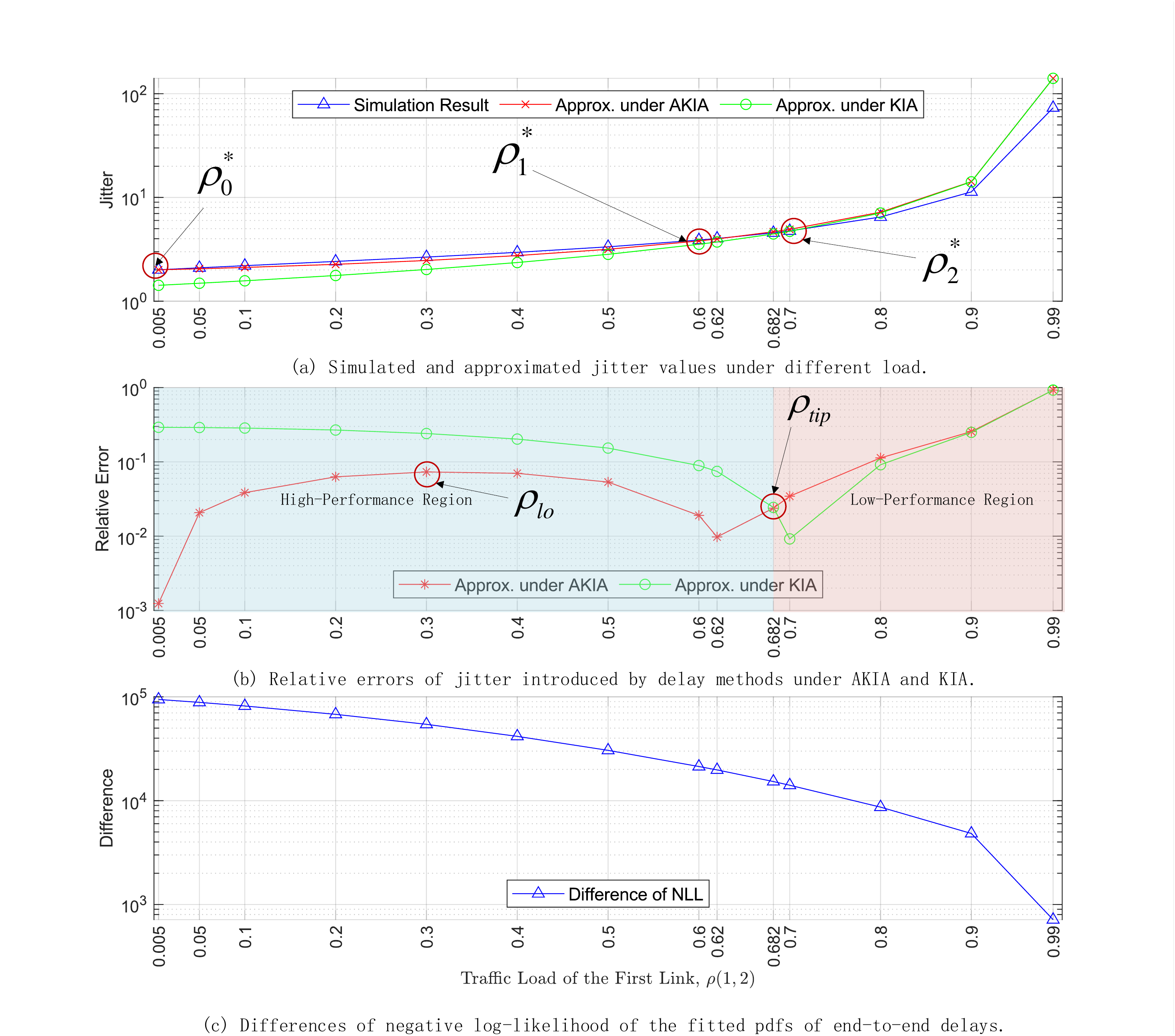}
\caption{Approximation performance comparison between the delay approximation methods under AKIA and KIA.}
\label{fig:approx_result_2h}
\end{figure*}

Based on \eqref{eq:wokia_mean-rm}, \eqref{eq:appr_jitter_ud_akia} and \eqref{eq:jitter_ud_kia}, the approximated average delay and jitter values under AKIA and KIA are:
\begin{equation}\label{eq:wokia_mean-rm_2h}
\tilde{\mu}_Z(\rho) = \tilde{\mu}_{\ddot{Z}}(\rho) = \frac{2}{\theta} = \frac{2}{1 - \rho},
\end{equation}
\begin{align}\label{eq:jitter_wo}\nonumber
  \tilde{\sigma}_{Z}(\rho)\approx&\sqrt{\frac{p_1(2-p_1)}{\theta^2_1}
  + \frac{p_2(2-p_2)}{\theta^2_2} + \frac{1}{\theta^2_3}}\\
  =& \sqrt{\frac{2\rho(2-\rho)}{(1-\rho)^2} + \frac{1}{(1/2)^2}}
\end{align}
and
\begin{equation}\label{eq:jitter_w_kia-1}
  \tilde{\sigma}_{\ddot{Z}}(\rho) \approx \sqrt{\sum\limits_{f = 1}^h {\frac{1}{\theta _f^2}}} = \sqrt{\frac{1}{\theta^2_2} + \frac{1}{\theta^2_2}}
  = \sqrt{\frac{2}{\theta^2}} = \sqrt{\frac{2}{(1-\rho)^2}},
\end{equation}
respectively. Moreover, $\tilde{\mu}_Z(\rho)$, $\tilde{\sigma}_Z(\rho)$ and
$\tilde{\sigma}_{\ddot{Z}}(\rho)$ are functions of traffic load $\rho$.

For $\tilde{\mu}_Z(\rho)$ in the $M/M/1- M/1$ system with equal service time,
Pinedo and Wolff \cite{Pinedo1982} conjectured that
there exists a $\rho^*$ such that for $\rho < \rho^*$, the expected delay will be longer than the
approximated delay, i.e., $\mu_Z(\rho) > \tilde{\mu}_Z$; and
$\rho > \rho^*$ the other way around. This conjecture was justified by simulation in their work \cite{Pinedo1982}.

For $\tilde{\sigma}_Z(\rho)$, I may also conjecture that
there also exists $\rho_1^*$ such that for $\rho < \rho_1^*$,
$\sigma_Z(\rho) > \tilde{\sigma}_Z(\rho)$; and
$\rho < \rho_1^*$ the other way around. There exists $\rho^*_2$ that has a similar
behavior for $\tilde{\sigma}_{\ddot{Z}}(\rho)$. Because $\tilde{\sigma}_{Z} > \tilde{\sigma}_{\ddot{Z}}$ (see
Remark \ref{rm:less}), we can infer that $\rho^*_1 < \rho^*_2$. Moreover, there should exist another
$\rho^*_0 = 0$ such that $\sigma_Z(0) = \tilde{\sigma}_Z(0)$ as $\rho \to 0$, we have $Z \to R$.

Fig. \ref{fig:approx_result_2h}(a) shows the simulated jitter values together with
approximated jitters values using $\tilde{\sigma}_Z(\rho)$
and $\tilde{\sigma}_{\ddot{Z}}(\rho)$ over traffic load $\rho$.
For each simulation, I assume that the values of the sample standard deviation (jitter) based on 100,000 samples adequately represent the true $\sigma_{Z}(\rho)$ values.
Fig. \ref{fig:approx_result_2h}(a) shows that
\begin{enumerate}
  \item the values of $\sigma_{Z}(\rho)$ are closely approximated by $\tilde{\sigma}_Z(\rho)$;
  \item $\tilde{\sigma}_Z(\rho) > \tilde{\sigma}_{\ddot{Z}}(\rho)$;
  \item $\rho^*_1 \approx 0.6$ and $\rho^*_2 \approx 0.7$.
\end{enumerate}


In each simulation run, we take the relative error to measure
the approximation accuracy of $\tilde{\sigma}_{Z}(\rho)$ and $\tilde{\sigma}_{\ddot{Z}}(\rho)$:
\begin{equation}\label{eq:rel_err}
\text { relative error }=\epsilon_{\tilde{\sigma}_{X}}(\rho) = \left|\frac{\tilde{\sigma}_{X}(\rho)- \sigma_{Z}(\rho)}{\sigma_{Z}(\rho)}\right|,
\end{equation}
where $X$ is a r.v. that can be $Z$ or $\ddot{Z}$. $\epsilon_{\tilde{\sigma}_{X}}(\rho)$
in \eqref{eq:rel_err} is a function of traffic load $\rho$ as well.
If $\rho^*_0$, $\rho^*_1$ and $\rho^*_2$ exist, then
we have $\epsilon_{\tilde{\sigma}_{Z}}(\rho^*_0) = \epsilon_{\tilde{\sigma}_{Z}}(0) = 0$, $\epsilon_{\tilde{\sigma}_{Z}}(\rho^*_1) = 0$ and
$\epsilon_{\tilde{\sigma}_{\ddot{Z}}}(\rho^*_2) = 0$. According to the intermediate value theorem,
I may conjecture that there exists
a $\rho_{tip} \in [\rho^*_1, \rho^*_2]$ such that if $\rho < \rho_{tip}$,
$\epsilon_{\tilde{\sigma}_{Z}}(\rho^*_1) < \epsilon_{\tilde{\sigma}_{\ddot{Z}}}(\rho^*_1)$;
and if $\rho < \rho_{tip}$, the reverse is true. Based on the above conjectures that $\epsilon_{\tilde{\sigma}_{Z}}(0) = 0$, $\epsilon_{\tilde{\sigma}_{Z}}(\rho^*_1) = 0$ and Rolle's theorem, I may further
conjecture that there exists
a $\rho_{lo} \in [0, \rho^*_1]$ such that $\epsilon_{\tilde{\sigma}_{Z}}(\rho_{tip})$
gives a poorest approximation with the highest relative error, i.e.,
\begin{equation}\label{xx}
\epsilon'_{\tilde{\sigma}_{Z}}(\rho_{tip}) = 0.
\end{equation}
Fig. \ref{fig:approx_result_2h}(b) shows the
relative errors introduced
by $\epsilon_{\tilde{\sigma}_{Z}}(\rho)$ and $\epsilon_{\tilde{\sigma}_{\ddot{Z}}}(\rho)$.
From this figure, it appears that $\rho_{tip}$ and $\rho_{lo}$ occur at $0.682$
($\epsilon_{\tilde{\sigma}_{Z}}=\epsilon_{\tilde{\sigma}_{\ddot{Z}}}= 0.0238$)
and $0.3$ ($\epsilon_{\tilde{\sigma}_{Z}}= 0.0732$ and
$\epsilon_{\tilde{\sigma}_{\ddot{Z}}}= 0.2408$), respectively. I call $(0,\rho_{tip})$ the high-performance region
$R_{hi}$ for $\epsilon_{\tilde{\sigma}_{Z}}(\rho)$
and $(\rho_{tip},1)$ the low-performance region $R_{lo}$
for $\epsilon_{\tilde{\sigma}_{Z}}(\rho)$. These two regions are show in Fig. \ref{fig:approx_result_2h}(b) as well. Moreover, when the traffic load approaches to one,
it can be observed that the relative errors introduced by both methods increase sharply,
e.g., $\epsilon_{\tilde{\sigma}_{Z}}(0.99) = 0.9246$ and
$\epsilon_{\tilde{\sigma}_{\ddot{Z}}}(0.99) = 0.9245$ when $\rho = 0.99$.

\subsubsection{Density Function Approximation Comparison}
\label{sec:densify_compari}

%

%
%

In such a two-node tandem network, we have 14 data sets generated from 14 different $\rho$ values. For each data set, we have 1 million delay samples $\{z_1, z_2, \cdots, z_{10^5}\}$. We compare the fitted delay density distributions under AKIA and KIA using the negative log-likehood (NLL), i.e.,
\begin{equation}\label{xx}
NLL(X) = - \sum_{i=1}^{10^6} \log L_X(z_i),
\end{equation}
where $X$ is a rv with a pdf and
$z_i$ is a delay sample. Let us precise that $\log L_Z(z_i)$ is the log-likelihood function (LL) using the pdf $\tilde{f}_Z(t)$
with $\mathbf{p}_Z$ of \eqref{eq:pz_2h_ud_akia} and $\boldsymbol{\theta}_Z$ of \eqref{eq:tz_2h_ud_akia}; and
$\log L_{\ddot{Z}}(z_i))$ is the LL using the pdf $\tilde{f}_{\ddot{Z}}(t)$
with $\boldsymbol{\theta}_W$ of \eqref{eq:tw_2h_ud_kia}. Smaller NLL means a better fit.
Denote by $\Delta$ the difference between NLL values under KIA and AKIA:
\begin{equation}\label{xx}
  \Delta = NLL(\ddot{Z}) -  NLL(Z),
\end{equation}
We see in Fig. \ref{fig:approx_result_2h}(c) that the differences for all traffic loads are all positive, suggesting the superiority of $\tilde{f}_Z(t)$ when PSNs have AKIA.

%
%
%
%
%

Finally, let us focus on the ccdfs when 1) $\tilde{\sigma}_Z$ has the worst performance in the
high-performance region $R_{hi}$ ($\rho = \rho_{lo} = 0.3$) and 2)
$\tilde{\sigma}_{\ddot{Z}}$ has the
best performance in the low performance region
$R_{lo}$ ($\rho = \rho_2 = 0.7$).
Figs. \ref{fig:two_node tandem_r_lo} and \ref{fig:two_node tandem_r_2}
shows the simulation and approximation results under these two $\rho$s.
The empirical ccdfs of end-to-end delays from simulation data, the approximated ccdfs of \eqref{eq:ph_cdf_delay}
(with $\mathbf{p}_Z$ and $\boldsymbol{\theta}_Z$ being \eqref{eq:pz_2h_ud_akia} and \eqref{eq:tz_2h_ud_akia}, respectively), and the ccdfs of \eqref{eq:ccdf_w_kia} (with $\boldsymbol{\theta}_W$ being \eqref{eq:tw_2h_ud_kia}) are plotted in blue, red and black dotted lines, respectively. The $y$-axis uses logarithmic scale. When $\rho = \rho_{lo}$, the simulation results show
the accuracy of $\tilde{S}_Z(t)$ in \eqref{eq:ph_cdf_delay}. On the other hand, the ccdf of
$\tilde{S}_{\ddot{Z}}(t)$ in \eqref{eq:ccdf_w_kia} significantly deviate from the simulation results. When $\rho = \rho_{lo}$, it is hard to differentiate the ccdfs $\tilde{S}_Z(t)$ and $\tilde{S}_{\ddot{Z}}(t)$. This is because, first, the compound queueing delay $W$ dominates the end-to-end delay, and second, the distribution functions of $W$ are the same from both delay approximation methods.

Based on the above results on goodness-of-fitness statistics and ccdfs, I conjecture that
\begin{conj}
In PSNs under AKIA, the delay approximation method under AKIA
significantly improves the approximation accuracy in the high-performance region. However,
both methods under AKIA and KIA have
similar approximation performance in the low-performance region. Moreover, their approximation may
significantly deviate from actual values.
\end{conj}

\begin{figure*}[t]
\centerline{\subfloat[$\rho = \rho_{lo}$ ($\tilde{\sigma}_{Z}(\rho_{lo})$ has the highest relative error in $R_{hi}$): empirical ccdf, approximated ccdfs under AKIA and KIA.]{\includegraphics[width=3.2in]{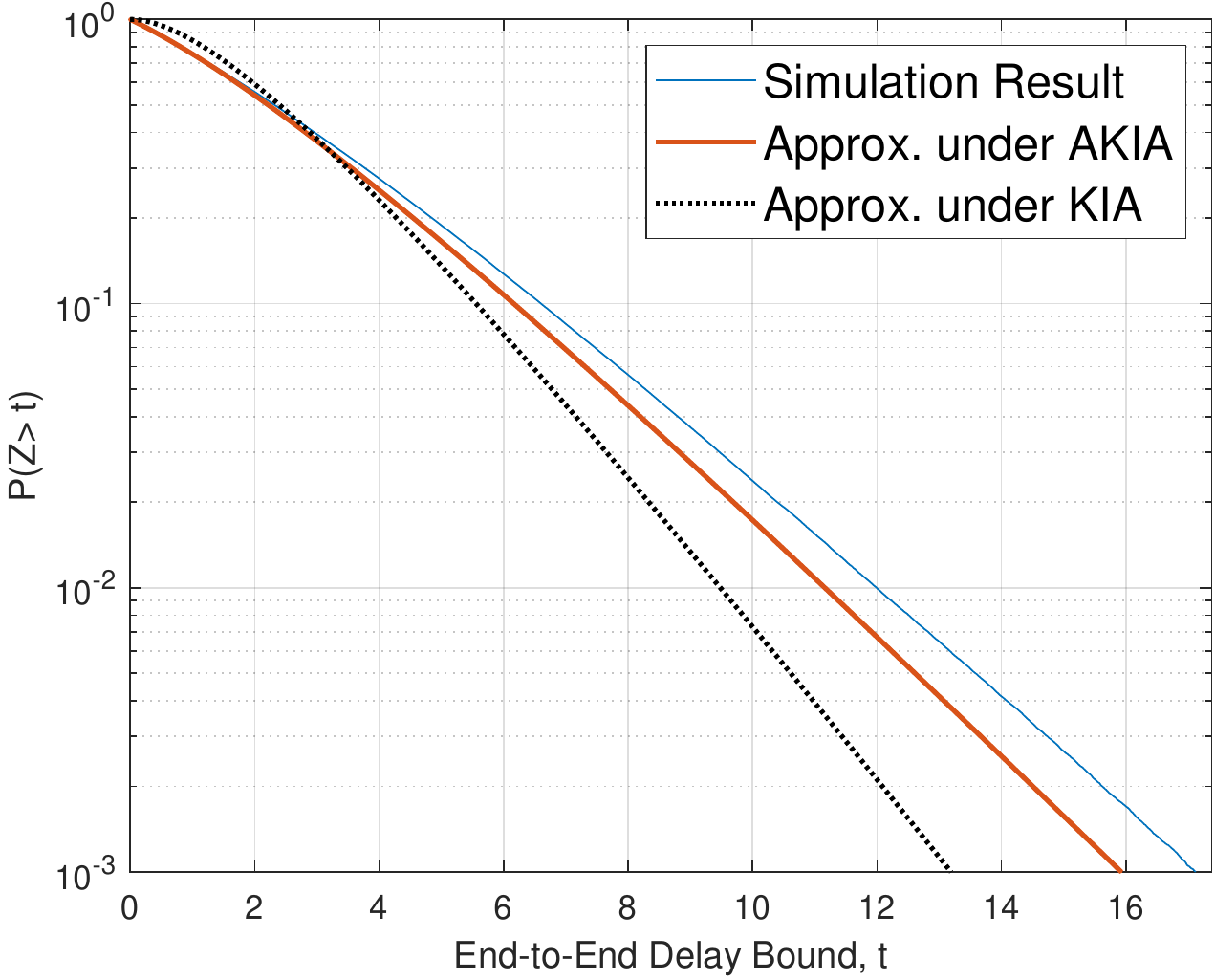}
\label{fig:two_node tandem_r_lo}}
\hfil
\subfloat[$\rho = \rho^*_{2}$ ($\tilde{\sigma}_{\ddot{Z}}(\rho_{2})$ has the lowest relative error in $R_{lo}$): empirical ccdf, approximated ccdfs under AKIA and KIA.]{\includegraphics[width=3.2in]{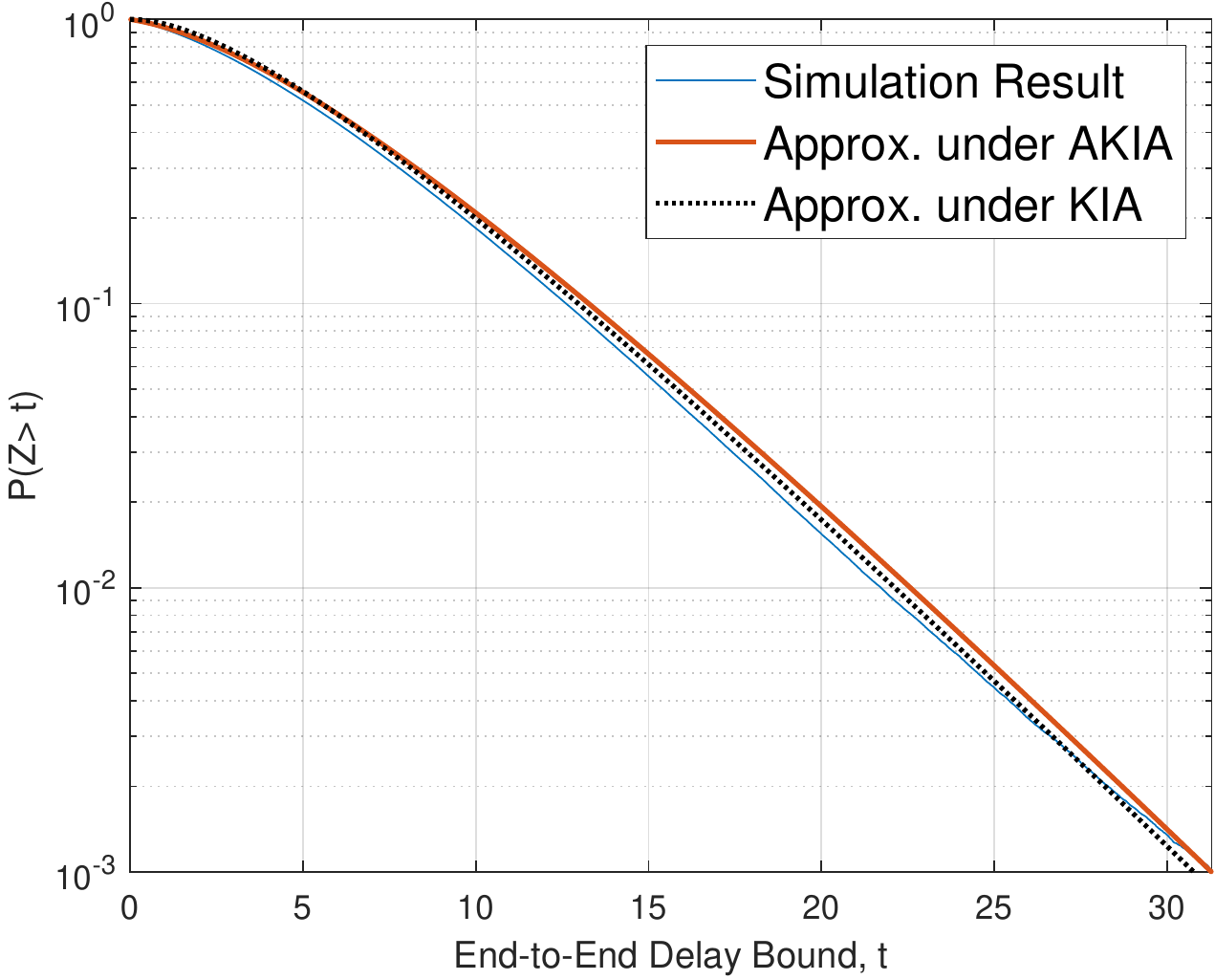}
\label{fig:two_node tandem_r_2}}}
\caption{Empirical and approximated ccdfs in a two-node tandem network.}
\label{fig:delay_perf}
\end{figure*}

%


\subsection{100-Node Network}
\label{sec:100-node}

\begin{figure*}[t]
\centering
\includegraphics[width=7in]{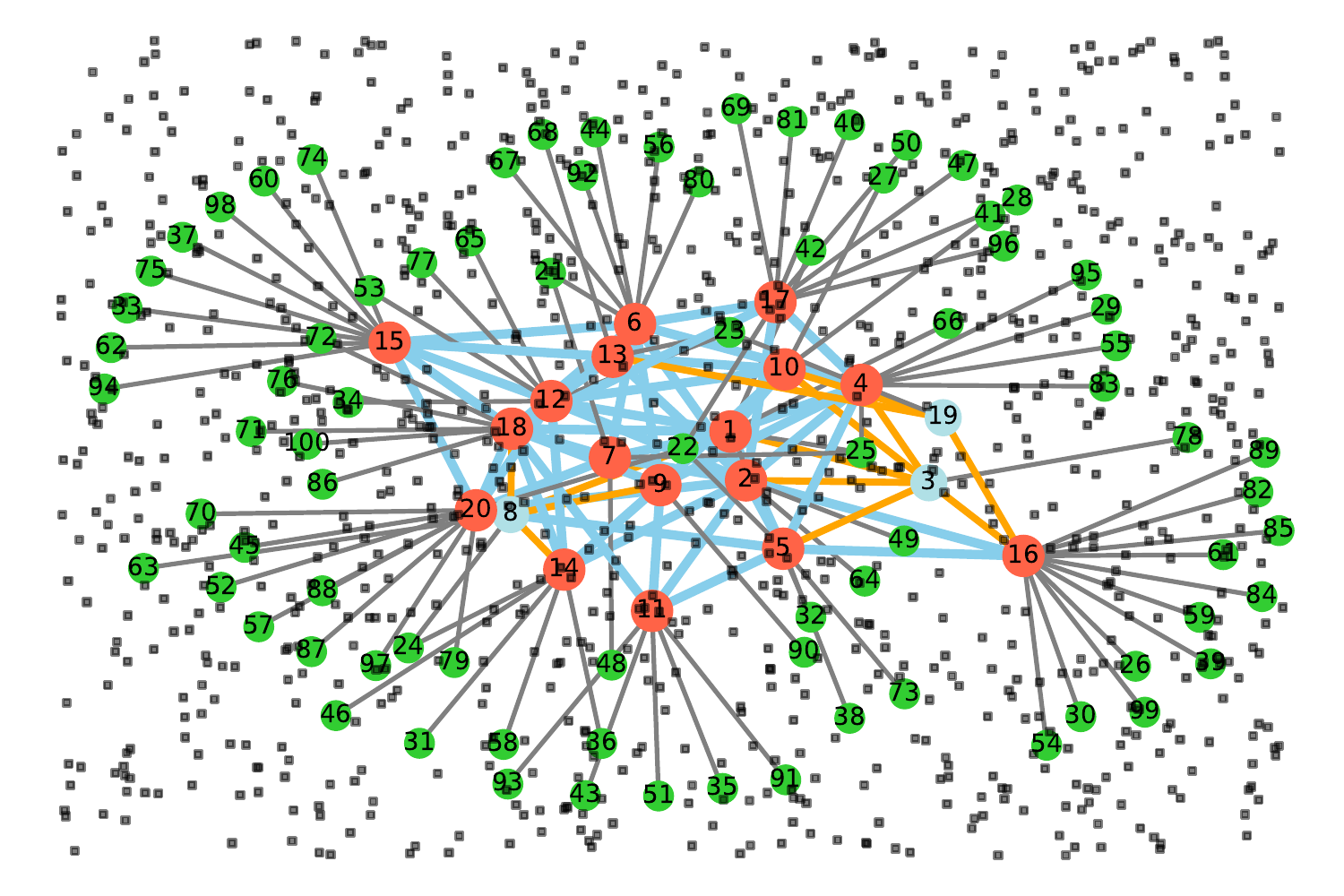}
\caption{Topology of a 100 nodes represented as circles and location of 1000 hosts represented as squares.}
\label{fig:100-node-topo}
\end{figure*}

\begin{figure*}[t]
\centering
\includegraphics[width=7in]{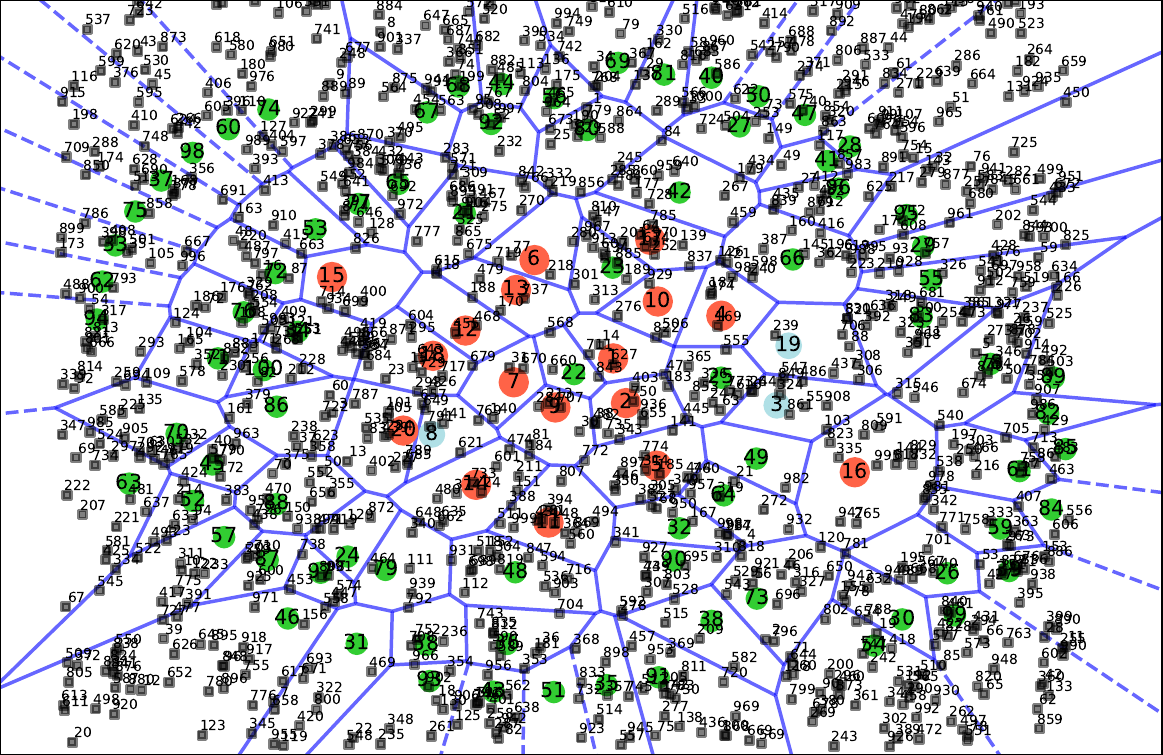}
\caption{Voronoi network topology; a Voronoi cell represents a subnet
that one node and a number of hosts are attached to.}
\label{fig:100-node-voro}
\end{figure*}


\begin{table}[]
\centering
\caption{Flows and their attributes}
\begin{tabular}{|r|p{0.4in}|p{0.4in}|p{0.4in}|p{0.4in}|p{0.6in}|}
\hline
\textbf{Flow}&\textbf{Source Host, $h_s$}&\textbf{Sink Host, $h_t$}
&\textbf{Source Node, $s$}&\textbf{Sink Node, $t$}&\textbf{Arrival Rate (pps), $\gamma$}\\\hline
1 	 & 	 1 	 & 	 2 & 80 & 73 & 1344.1 \\\hline
2 	 & 	 2 	 & 	 1 & 73 & 80& 1344.1\\\hline
3 	 & 	 3 	 & 	 4 & 86 & 32& 1344.1\\\hline
4 	 & 	 4 	 & 	 3 & 32 & 86& 1344.1\\\hline
$\cdots$ 	 & 	 $\cdots$ 	 & 	 $\cdots$ &$\cdots$&$\cdots$&$\cdots$ \\\hline
127 	 & 	 127 	 & 	 128 & 60 & 77 & 1344.1 \\\hline
$\cdots$ 	 & 	 $\cdots$ 	 & 	 $\cdots$ &$\cdots$&$\cdots$&$\cdots$ \\\hline
635 	 & 	 635 	 & 	 636 & 14 & 83 & 1344.1 \\\hline
$\cdots$ 	 & 	 $\cdots$ 	 & 	 $\cdots$ &$\cdots$&$\cdots$&$\cdots$ \\\hline
997 	 & 	 997 	 & 	 998 & 92 & 32  & 1344.1\\\hline
998 	 & 	 998 	 & 	 998 & 32 & 92 & 1344.1\\\hline
999 	 & 	 999 	 & 	 1000 & 11 & 40 & 1344.1\\\hline
1000 	 & 	 1000 	 & 	 999 & 40 & 11 & 1344.1\\\hline
\end{tabular}
\label{tab:flows}
\end{table}


Let us simulate a complex computer network that consists of a number of routers and hosts in this section.
A 100-node/router hypothetical computer network is considered; its topology $G(V, E)$ is a directed graph resembling an Internet Autonomous System (AS) network.
I used a built-in \emph{random\_internet\_as\_graph} function in the networkx toolbox \cite{Hagberg2008} to generate such a topology; the function is
implemented based on the work \cite{Elmokashfi2010}. There are three types of routers: red and blue
circles represent 10 Gbps and 1 Gbps core routers, green circles represent 100 Mbps edge routers.
I also let 1000 hosts locate uniformly
in the same area. The network topology of routers and the locations of hosts are shown in
Fig. \ref{fig:100-node-topo} The central routing directory $R$ is constructed by Dijkstra's algorithm.

A Voronoi tessellation is given in Fig. \ref{fig:100-node-voro}.
I use the term \emph{subnet} to refer to a Voronoi region of each router. Therefore, hosts in a Voronoi region are either directly connected or on the same subnet with a router in the same Voronoi region. One host is communicating with another
host (we have 500 host pairs or 1000 traffic flows) with a data rate of 2 Mbps. Because it was reported in \cite{Thompson1997} that the mean size of IP packets is 186 bytes, the same value is used for the average packet length $1/\mu$ and the arrival rate of any flow $\gamma=2 \times 10^6 / (186 \times 8)=1344.1$
packets per second (pps). Traffic flows of different host pairs with their associated routers and
arrival rate are listed in Table \ref{tab:flows}.

The simulation stops when the simulation time passes
50 seconds. That is, the simulation platform generated about $50 \times 1344.1 = 67200 $ end-to-end delay samples for every flow, so in total we have $67200 \times 1000 = 67.2$ million
delay samples. 964 flows with path length $h > 1$ were analyzed. 36 flows were not considered because either their path length is one, resulting the same delay approximation under AKIA and KIA (see Remark \ref{rm:one-hop}) or their source and destination hosts are in the same subnet, leaving their end-to-end delays all zero.

Fig. \ref{fig:relative_errors_100} shows the
relative errors introduced
by $\epsilon_{\tilde{\sigma}_{Z}}$ and $\epsilon_{\tilde{\sigma}_{\ddot{Z}}}$ for 964 flows in log scale.
I found 919 flows in the high-performance region and 45 flows in the low-performance region.
In other words, 95.33\% ($919/964$) of jitter approximation made by $\tilde{\sigma}_{Z}$ outperforms
$\tilde{\sigma}_{\ddot{Z}}$. Moreover,
$\tilde{\sigma}_{Z}$ introduces relative errors
in the range of $6.74\times 10^{-6}$ to $0.026$ while
$\tilde{\sigma}_{\ddot{Z}}$ introduces relative errors
in the range of $2.02\times 10^{-4}$ to $0.314$. In such a simulation setting,
we may conclude $\tilde{\sigma}_{Z}$
significantly reduces the approximation error by 12.08 times.

We again compare the fitted delay distributions using the difference between NLL under KIA and AKIA. We see in Fig. \ref{fig:nll_100} that differences are positive for the 964 traffic flows with no exceptions.
That means the distribution using the delay approximation method under AKIA always has smaller NLL values so
$\tilde{S}_Z(t)$ better fits the empirical distributions.
Finally, let us focus on the ccdfs when 1) $\tilde{\sigma}_Z$ has the highest relative error in the
high-performance region $R_{hi}$ (flow $127$: $\epsilon_{\tilde{\sigma}_Z} = 0.026$,
$\epsilon_{\tilde{\sigma}_{\ddot{Z}}} = 0.2199$) and 2) $\tilde{\sigma}_{\ddot{Z}}$ has the lowest
relative error in the low performance region
$R_{lo}$ (flow $635$: $\epsilon_{\tilde{\sigma}_Z} = 0.010$,
$\epsilon_{\tilde{\sigma}_{\ddot{Z}}} = 2.02 \times 10^{-4}$).
Other information of flows $127$ and $635$ is shown in Table \ref{tab:flows}.
Figs. \ref{fig:ccdf_r1_100} and \ref{fig:ccdf_r2_100}
show the simulation and approximation results for these two flows.
It is shown in the figures that the
simulation and the approximation results using $\tilde{S}_Z(t)$ are in good agreement for these flows even $\tilde{\sigma}_Z$ made the worst approximation for them.
On the other hand, we observe a clear
difference between $\tilde{S}_{\ddot{Z}}(t)$ and the empirical ccdf in Fig. \ref{fig:ccdf_r1_100}
but such a difference becomes unnoticeable in Fig. \ref{fig:ccdf_r2_100}.
This observation confirms the conjecture I made in Section \ref{sec:densify_compari}.

\subsection{Effect of the Service Time Dependence}

By investigating the average packet delay via extensive simulation\footnote{Interested readers can refer to Section 3.5 of \cite{Kleinrock1964} for results on richer topologies.}, Kleinrock argued that \cite{Kleinrock1964}:
\begin{klein}
If there is sufficient mixing of traffic, then
the dependence effect may be small, resulting in an effect of restoring the independence of interarrival times and packet lengths.
\end{klein}

Kleinrock's observation is partially justifiable. Because the approximated average delays using \eqref{eq:wokia_mean-rm} and \eqref{eq:wokia_mean-1} are equal (a result of Remark \ref{rm:same_avg}),
On the other hand,
the highest relative error in the high-performance region in a two-node tandem network of Fig.\ref{fig:2-node} is $0.0732$ (when $\rho=0.3$)
while the highest error in the 100-node network in high- and low-performance regions is $0.027$ for flow $127$. The approximated
ccdf of flow $127$ using $\hat{S}_Z(t)$ is clearly much closer to the empirical ccdf than the
approximated ccdf when $\rho=0.3$ in a two-node tandem network. Therefore, it may be more appropriate to revise Kleinrock's observation as follows:
\begin{klein_rv}
If there is sufficient mixing of traffic, then
the dependence effect may be small, resulting in an effect of restoring the independence of interarrival times \st{and packet lengths}.
\end{klein_rv}
The above observation serves as the basis for the assumption A1 I made in Section \ref{sec:packet}. Moreover, the assumptions A2 and A3 may also be reasonably valid based on this observation.

\begin{figure*}[t]
\centering
\includegraphics[width=7in]{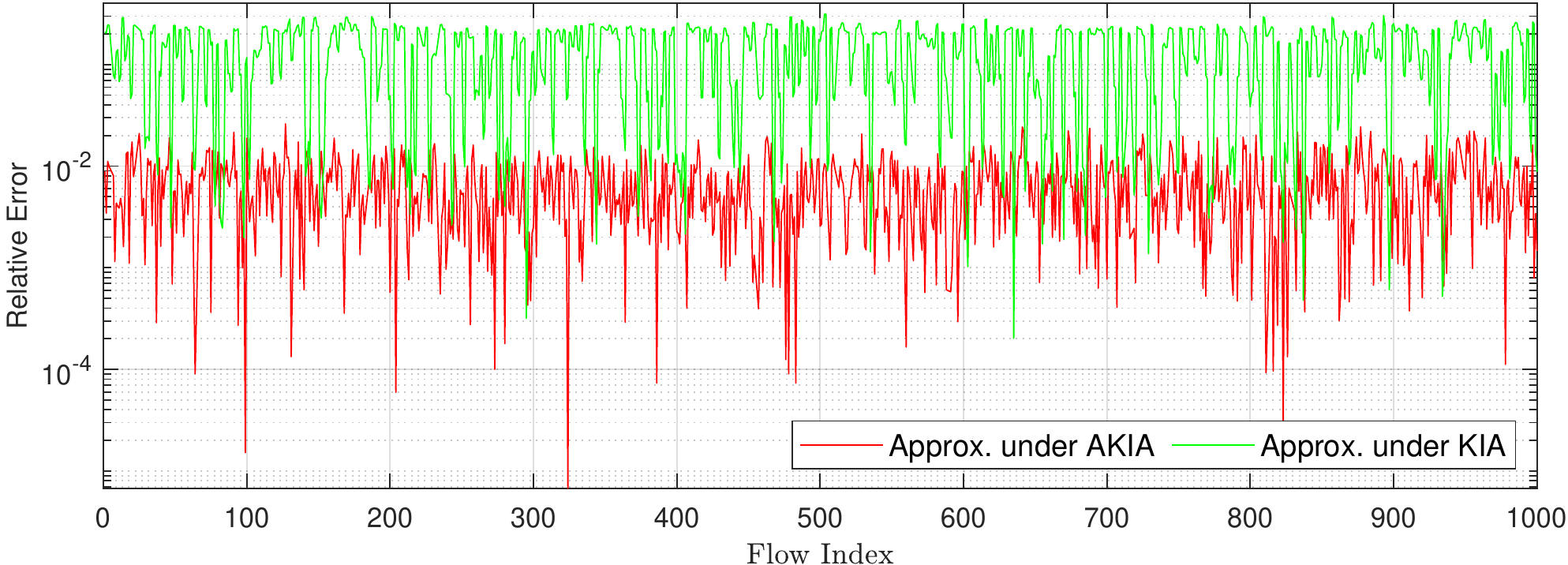}
\caption{Relative errors introduced by the approximation method under AKIA and KIA for 964 flows.}
\label{fig:relative_errors_100}
\end{figure*}

\begin{figure*}[t]
\centering
\includegraphics[width=7in]{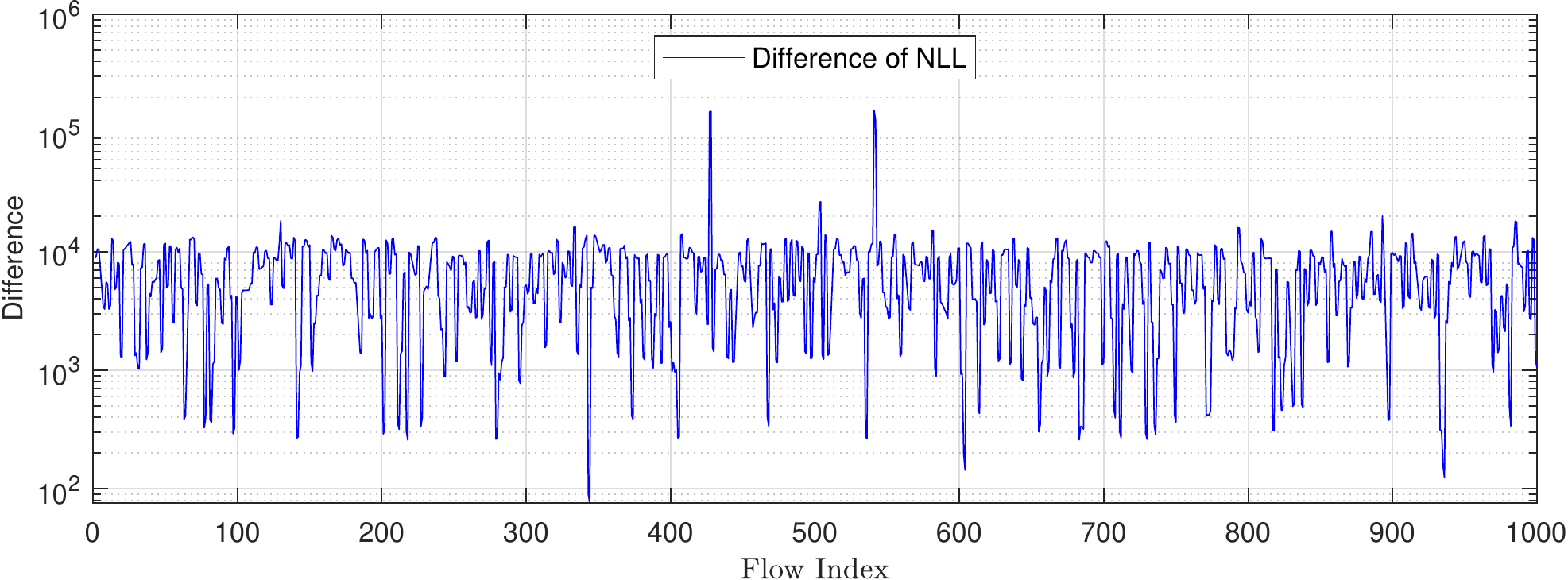}
\caption{Difference of the negative log-likelihood of the fitted pdfs $\tilde{f}_Z(t)$ and $\tilde{f}_{\ddot{Z}}(t)$ of end-to-end delays.}
\label{fig:nll_100}
\end{figure*}

\begin{figure*}[t]
\centerline{
\subfloat[Flow 127 ($\tilde{\sigma}_{Z_{127}}$ has the highest relative error in $R_{hi}$): empirical ccdf, approximated ccdfs under AKIA and KIA.]{\includegraphics[width=3.2in]{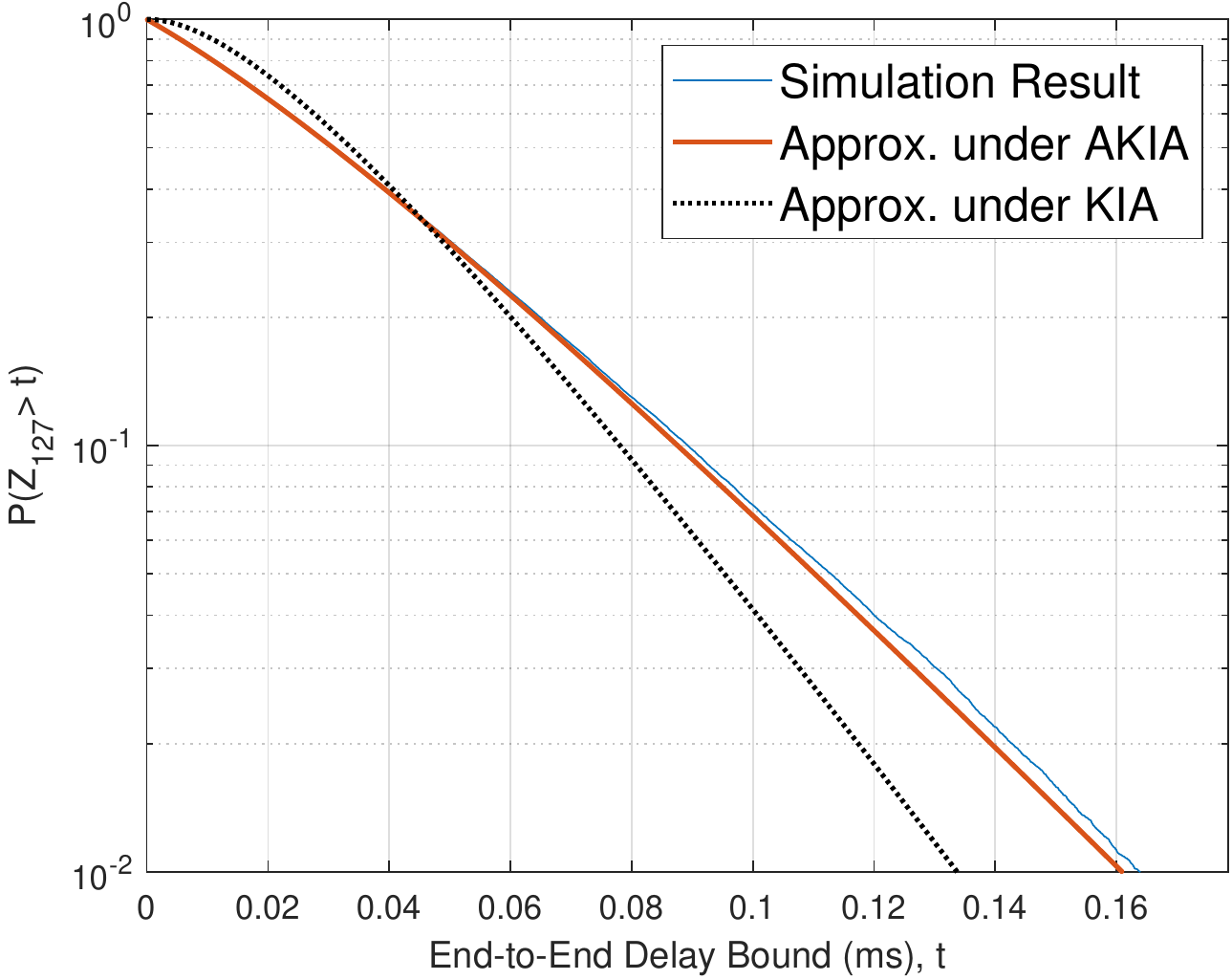}
\label{fig:ccdf_r1_100}}
\hfil
\subfloat[Flow 635 ($\tilde{\sigma}_{\ddot{Z}_{127}}$ has the lowest relative error in $R_{lo}$): empirical ccdf, approximated ccdfs under AKIA and KIA.]{\includegraphics[width=3.2in]{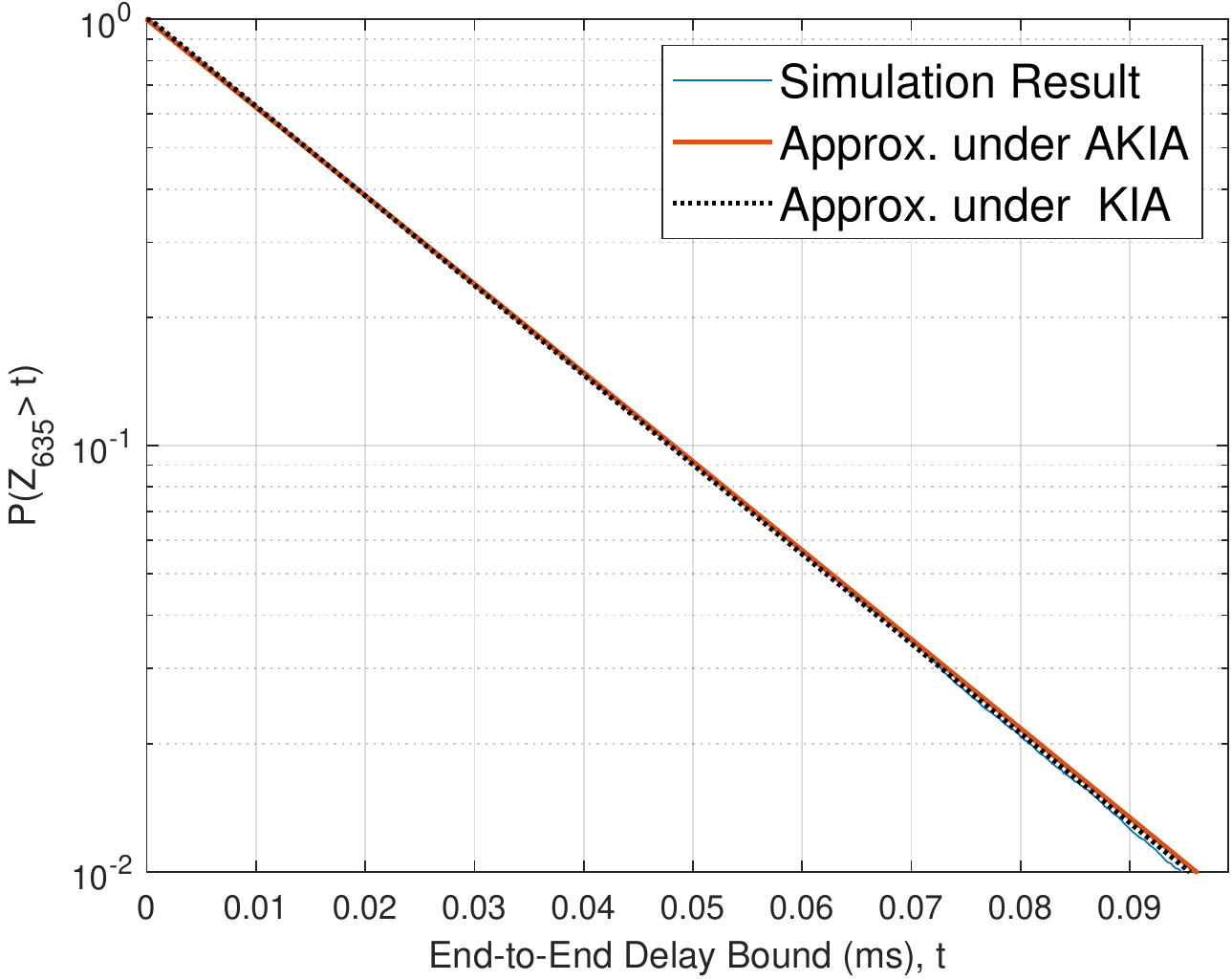}
\label{fig:ccdf_r2_100}}}
\caption{Empirical and approximated ccdfs in a 100-node network with 1000 traffic flows.}
\label{fig:delay_perf}
\end{figure*}

\section{Conclusion}

In this paper, I give an answer to the question on the end-to-end delay approximation in packet-switched networks
when packet lengths are kept unchanged when they traverse across networks; the question was raised by Kleinrock in 1961. In particular, I identify a new phase-type distribution, which I derived in my previous work, as a new phase-type distribution.
Using this new distribution, I bridge some results from the network flow theory and the queueing
theory; and then propose a generalized approximation method for the end-to-end delay (average delay, jitter and delay
distribution) of every stochastic flow in queueing networks with dependent or independent service
time of any size. 
Simulation results confirm the accuracy of this new approximation method.
Finally, I observe that the influence of the correlations due to dependent service times can be significantly reduced in networks with a larger number of nodes and complex traffic mixing.
This observation suggests the practicality of my model in real and complex networks.


\appendices

\section{Properties of $C(\mathbf{p}, \boldsymbol \theta)$}
\label{sec:prop}


%

\begin{figure*}[t]
\centering
\includegraphics[width=6in]{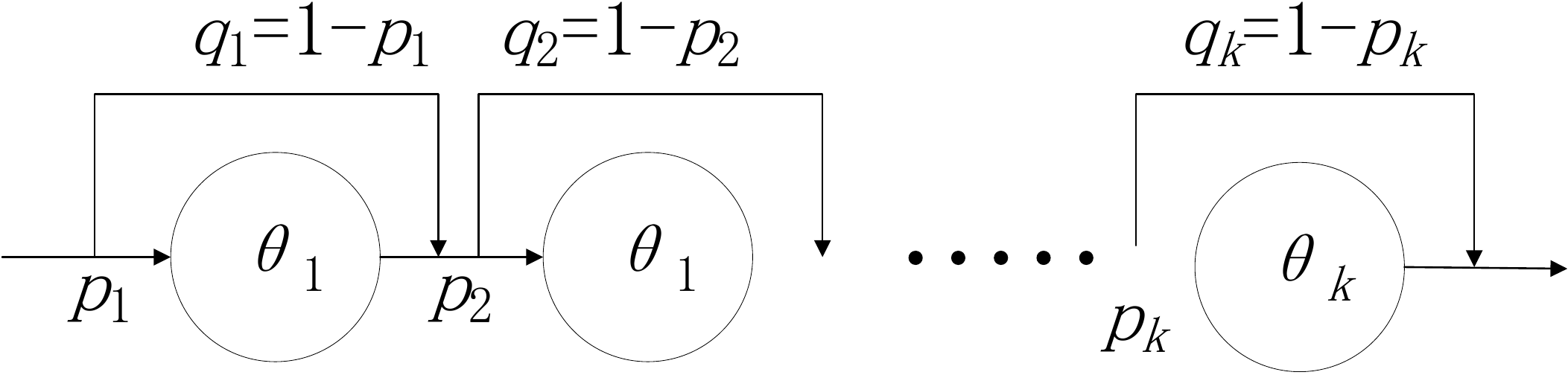}
\caption{$h$-phase Markov chain for $C(\mathbf{p},\boldsymbol \theta)$ with $\mathbf{p} = \{p_1, p_2, \cdots, p_h\}$ and $\boldsymbol \theta = \{\theta_1, \theta_2, \cdots, \theta_h\{$.}
\label{fig:markov}
\end{figure*}


Since a phase-type distribution is defined as the time to absorption in a specified continuous-time
Markov chain, I demonstrate another method to generate the rv $Z=\sum_{f=1}^h U_{f}$. Consider a $h$-phase Markov chain in Fig. \ref{fig:markov}. This Markov chain has
the property: the time in phase $f~(f \le h)$ is either exponential with mean $1/\lambda_f$ (with probability $p_f$) or zero (with probability $q_f = 1 - p_f$).
Let $Z$ denote the time until the process leaves the system. Then $Z$ is a phase-type distributed rv with a pdf of \eqref{eq:ph_pdf}.

Remark \ref{pro:deg_exp_erl} below is used in Sections \ref{sec:wo_kia} and \ref{sec:w_kia}:

\begin{remark}\label{pro:deg_exp_erl}
$C(\mathbf{p}, \boldsymbol \theta)$ generalizes the 1) degenerate with point
mass at zero, 2) exponential, 3) hypoexponential and 4) Erlang distributions.
\end{remark}

It is straightforward to show the generalization property of $C(\mathbf{p}, \boldsymbol \theta)$ in the above remark. $C(\mathbf{p}, \boldsymbol \theta)$ will be in the form of a degenerate distribution with point mass at zero if $h=1$ and $p_1 = 0$; it will be an exponential distribution if $h=1$ and $p_1 = 1$; it will be a hypoexponential distribution if $h\ge 2, \forall p_f = 1 (f \le h)$ and ${\theta{}}_f\neq{\lambda{}}_g (f\neq g, \forall f, g \le h)$; it will be the Erlang distribution if $h\ge 2, \forall p_f = 1 (f \le h)$ and ${\theta{}}_f={\lambda{}}_g (f\neq g, \forall f, g \le h)$.

\section{Average End-to-End Delay \eqref{eq:wokia_mean-rm}}
\label{sec:proof_avg}

According to \eqref{eq:12-mean}, I have
\begin{align}\nonumber\label{a}
E[Z_{st}] &\mathop  = \limits^{(1)} \sum\limits_{f = 1}^{h+1} {\frac{{{p_f}}}{{\theta_f}}} \mathop  = \limits^{(2)} \sum\limits_{f = 1}^{h} {\frac{{{p_f}}}{{\theta_f}}} + \frac{1}{\mu \overline{c}} \mathop  = \limits^{(3)} \sum\limits_{f = 1}^{h}\left({\frac{{{p_f}}}{{\theta_f}}} + \frac{1}{\mu c_{s_f s_{f+1}}}\right)\\\nonumber
& \mathop  = \limits^{(4)} \sum\limits_{f = 1}^{h}\left({\frac{{\lambda_{s_f s_{f+1}}}}{\mu c_{s_f s_{f+1}} }} \frac{1}{\mu c_{s_f s_{f+1}}-\lambda_{s_f s_{f+1}}}+ \frac{1}{\mu c_{s_f s_{f+1}}}\right)\\
&\mathop  = \limits^{(5)} \sum\limits_{f = 1}^h {\frac{{{1}}}{{\mu_{s_f s_{f+1}}C - \lambda_{s_f s_{f+1}}}}}
\end{align}
The $3^{\rm rd}$ equality holds because of \eqref{eq:c_ij_sk}.


\section{Proof of Remark \ref{rem:exp}}
\label{sec:proof_exp}

Let us have a lemma below:

\begin{lemma}\label{rm:another_v}
  Let $Z_1 \sim C(p_1, \theta_1)$ and $Z_2 \sim C(p_2, \theta_2)$. If
  the following condition holds:
  \begin{equation}\label{eq:cond}
    \theta_2 - \theta_1 = p_1 \theta_2,
  \end{equation}
then $Z=Z_1+Z_2$ and $Z \sim C(p, \theta)$ with
\begin{equation}\label{eq:cond_p}
  p={{ \frac{{{\Delta p_1 + p_1 p_2}{\theta _1}}}{\Delta}} }.
\end{equation}
and
\begin{equation}\label{eq:cond_theta}
\theta=\theta_1.
\end{equation}

%
%
%
%
\end{lemma}

\begin{proof}

The ccdf of $Z~(Z_1+Z_2)$ is
\begin{equation}\label{eq:rm10}
 {S_{Z}}(t) = \left(1+ {\frac{{ {p_2}{\theta _1}}}{{{\theta _2} - {\theta
_1}}}} \right){p_1}{e^{ - {\theta _1}t}} +\left(1+ {\frac{{{p_1}{\theta
_2}}}{{{\theta _1} - {\theta _2}}}} \right){p_2}{e^{ - {\theta _2}t}}.
\end{equation}
Let $\Delta = \theta_2 - \theta_1$. Eq. \eqref{eq:rm10} can be rewritten as
\begin{align}\nonumber
&= \left( {1 + \frac{{{p_2}{\theta _1}}}{\Delta }} \right){p_1}{e^{ - {\theta _1}t}} + \left( {1 + \frac{{{p_1}\left( {{\theta _1} + \Delta } \right)}}{{ - \Delta }}} \right){p_2}{e^{ - \left( {{\theta _1} + \Delta } \right)t}}\\\label{eq:prof_3}
&= {e^{ - {\theta _1}t}}\left( {\left( {1 + \frac{{{p_2}{\theta _1}}}{\Delta }} \right){p_1} + \left( {1 + \frac{{{p_1}\left( {{\theta _1} + \Delta } \right)}}{{ - \Delta }}} \right){p_2}{e^{ - \Delta t}}} \right)
\end{align}
If I arbitrarily let $1 + \frac{{{p_1}\left( {{\theta _1} + \Delta } \right)}}{{ - \Delta }}$ in \eqref{eq:prof_3} be zero, i.e.,
\begin{align}\nonumber
 &1 + \frac{{{p_1}\left( {{\theta _1} + \Delta } \right)}}{{ - \Delta }}=0
 \Longleftrightarrow  \Delta = {{p_1}\left( {{\theta _1} + \Delta } \right)}\\\label{eq:cond_pro}
& \Longleftrightarrow  \theta_2 - \theta_1 = p_1 \theta_2,
\end{align}
then \eqref{eq:prof_3} can be written as
\begin{equation}\label{eq:prof_3_2}
S_{Z}(t) = p{e^{ - {\theta _1}t}},
\end{equation}
where
\begin{equation}\label{eq:cond_pro}
p={{ \frac{{{\Delta p_1 + p_1 p_2}{\theta _1}}}{\Delta}} }.
\end{equation}
Eq. \eqref{eq:cond_pro} is the condition of \eqref{eq:cond} and eq. \eqref{eq:prof_3_2} the ccdf of
$Z \sim C(p,\theta)$ with $p$ of \eqref{eq:cond_p} and $\theta$ of \eqref{eq:cond_theta}.

\end{proof}

The above lemma suggests that under a condition of \eqref{eq:cond}, the sum of two independent $C(\mathbf{p}, \boldsymbol \theta)$ distributed rvs with $h=1$ is another $C(\mathbf{p}, \boldsymbol \theta)$ distributed rv with $h=1$.


The result in Remark \ref{rem:exp} can be obtained from lemma \ref{rm:another_v} because $W(x_j,x_{j+1})$ and $\ddot{R}(x_j,x_{j+1})$ satisfy the condition $\theta_2 - \theta_1 = p_1 \theta_2$ of Remark \ref{rm:another_v}. In particular, let us temporarily let $\theta_1 = \mu c(x_j,x_{j+1}) - \lambda(x_j,x_{j+1})$,
$\theta_2 = \mu c(x_j,x_{j+1})$,
$p_1 = \lambda(x_j,x_{j+1})/\mu c(x_j,x_{j+1})$ and $p_2=1$. We can verify that the condition
$\theta_2 - \theta_1 = p_1 \theta_2$ holds. Moreover,
let $\Delta = \theta_2 - \theta_1 = \lambda(x_j,x_{j+1})$, I have
\begin{equation}\label{a}
p={{ \frac{{{\Delta p_1 + p_1 p_2}{\theta _1}}}{\Delta}} } = \frac{\Delta p_1 + p_1 \theta_1}{\Delta}=\frac{p_1\theta_2}{\Delta}=1.
\end{equation}

\section{Proof of Remark \ref{rm:less}}
\label{sec:less}

Based on \eqref{eq:appr_jitter_ud_akia} and \eqref{eq:jitter_ud_kia}, we have
\begin{align}\label{a}\nonumber
&\tilde{\sigma}^2_{{Z}} - \tilde{\sigma}^2_{\ddot{Z}}=   Var\left[\sum_{j=1}^{h} R(x_j,x_{j+1})\right] - Var\left[\sum_{j=1}^{h} \ddot{R}(x_j,x_{j+1})\right] \\\nonumber
&= \frac{1}{\overline{c}^2}Var[V] -  \left(\sum_{j=1}^{h} \frac{1}{c^2({x_j, x_{j+1}})}\right)Var[V(x_j,x_{j+1})]\\\nonumber
&= \left(\sum_{j=1}^{h} \frac{1}{c(x_j,x_{j+1})}\right)^2Var[V] -  \\\nonumber
&~~~\left(\sum_{j=1}^{h} \left(\frac{1}{c(x_j,x_{j+1})}\right)^2\right)Var[V(x_j,x_{j+1})]\\\nonumber
&= \left(\left(\sum_{j=1}^{h} \frac{1}{c(x_j,x_{j+1})}\right)^2 -  \sum_{j=1}^{h} \left(\frac{1}{c(x_j,x_{j+1})}\right)^2\right)Var[V(x_j,x_{j+1})]
\end{align}
Based on the Cauchy–Schwarz inequality, we have the following inequality:
\begin{equation}\label{a}
\left(\sum_{j=1}^{h} \frac{1}{c(x_j,x_{j+1})}\right)^2 > \sum_{j=1}^{h} \left(\frac{1}{c(x_j,x_{j+1})}\right)^2
\end{equation}
Therefore, we have
\begin{align}\label{eq:der_diff_sig}
\sigma^2_{{Z}} - \sigma^2_{\ddot{Z}}>0
\Longleftrightarrow \sigma_{{Z}} > \sigma_{\ddot{Z}}
\end{align}
which is \eqref{eq:std-wo-w}. The ``$\Leftrightarrow$'' in \eqref{eq:der_diff_sig} is an equivalence sign.



%


\ifCLASSOPTIONcaptionsoff
  \newpage
\fi

\bibliographystyle{IEEEtran}
\bibliography{part2}

\end{document}